\documentclass[nopreprintline]{elsarticle}
\usepackage[utf8]{inputenc}
\usepackage{amsmath}
\usepackage{amsthm}
\usepackage{amsfonts}
\usepackage{amssymb}
\usepackage{float}
\usepackage{numbertabbing}
\usepackage{multicol}
\usepackage{ebproof}
\usepackage{multicol}

\usepackage{tikz}
\usetikzlibrary{positioning}
\usetikzlibrary{shapes}
\usetikzlibrary{automata}
\usetikzlibrary{backgrounds}
\usetikzlibrary{fit}
\usetikzlibrary{decorations.pathmorphing}
\usetikzlibrary {arrows.meta}
\usepackage{xspace}
\usepackage[shortlabels]{enumitem}

\usetikzlibrary{arrows,calc,patterns,positioning,shapes}
   \tikzset{
       modal/.style={>=stealth',shorten >=1pt,shorten <=1pt,auto,
                     node distance=1.5cm,semithick},
       world/.style={circle,draw,minimum size=1cm,fill=gray!15},
       point/.style={circle,draw,fill=black,inner sep=0.5mm},
       reflexive/.style={->,in=120,out=60,loop,looseness=#1},
       reflexive/.default={5},
       reflexive point/.style={->,in=135,out=45,loop,looseness=#1},
       reflexive point/.default={25},
       }

\floatstyle{ruled}
\newfloat{algo}{htbp}{algo}
\floatname{algo}{Construction}

\newcommand{\BC}{\ensuremath{\mathbb{C}}\xspace}

\newcommand{\BN}{\ensuremath{\mathbb{N}}\xspace}

\newcommand{\CC}{\ensuremath{\mathcal{C}}\xspace}

\newcommand{\CL}{\ensuremath{\mathcal{L}}\xspace}
\newcommand{\CM}{\ensuremath{\mathcal{M}}\xspace}
\newcommand{\CN}{\ensuremath{\mathcal{N}}\xspace}

\newcommand{\CS}{\ensuremath{\mathcal{S}}\xspace}

\newcommand\ag{ \mathsf{Ag} \xspace}
\newcommand\syn{ \mathsf{Syn} \xspace}
\newcommand\synminus{ \mathsf{Syn^{\text{-}}} \xspace}
\newcommand\pow{ \mathsf{Pow} \xspace}
\newcommand\alive{ \mathsf{Alive} \xspace}
\newcommand\lalive{ \mathsf{alive} \xspace}
\newcommand\dead{ \mathsf{dead} \xspace}

\newcommand\agsi{ \mathsf{Agsi} \xspace}

\newcommand\si{ \mathsf{S}1 \xspace}
\newcommand\sii{ \mathsf{S}2 \xspace}
\newcommand\siii{ \mathsf{S}3 \xspace}
\newcommand\kone{ \mathsf{K}1 \xspace}
\newcommand\ktwo{ \mathsf{K}2 \xspace}
\newcommand\kthree{ \mathsf{K}3 \xspace}
\newcommand\kfour{ \mathsf{K}4 \xspace}
\newcommand\dsf{ \mathsf{D} \xspace}

\renewcommand{\P}{\mathsf{Prop} \xspace}
\newcommand\ti{ \mathsf{T}1 \xspace}
\newcommand\tii{ \mathsf{T}2 \xspace}

\newcommand\itemktwo{\item[$\mathsf{K}2$:]}
\newcommand\itemkthree{\item[$\mathsf{K}3$:]}
\newcommand\itemkfour{\item[$\mathsf{K}4$:]}
\newcommand\itemd{\item[$\mathsf{D}$:]}
\newcommand\itemne{\item[$\mathsf{NE}$:]}

\newtheorem{theorem}{Theorem}
\newtheorem{corollary}{Corollary}
\newtheorem{definition}{Definition}
\newtheorem{remark}{Remark}
\newtheorem{lemma}{Lemma}
\newtheorem{example}{Example}

\begin{document}
\title{Synergistic Knowledge\tnoteref{t1,t2}}
\tnotetext[t1]{This is a revised and extended version of a preliminary work presented at the 25th International Symposium on Stabilizing, Safety, and Security of Distributed Systems (SSS), Jersey City, NJ, October 2023~\cite{DBLP:conf/sss/CachinLS23}. An extended abstract was presented at 5th International Conference on Logic and Argumentation (CLAR), Hangzhou, China, September 2023~\cite{DBLP:conf/clar/CachinLS23}.}
\tnotetext[t2]{This work has been funded by the Swiss National Science Foundation (SNSF)
under grant agreement Nr\@.~200021\_188443 (Advanced Consensus Protocols).}
\author[1]{Christian Cachin}
\ead{christian.cachin@unibe.ch}
\affiliation[1]{organization = {University of Bern}, addressline = {Neubrückstrasse 10}, postcode = {3012 CH}, city = {Bern}, country = {Switzerland}}

\author[1]{David Lehnherr}
\ead{david.lehnherr@unibe.ch}

\author[1]{Thomas Studer}
\ead{thomas.studer@unibe.ch}
\date{\today}            
\begin{abstract}
In formal epistemology, group knowledge is often modelled as the knowledge that the group would have, if the agents shared all their individual knowledge. However, this interpretation does not account for relations between agents. In this work, we propose the notion of synergistic knowledge which makes it possible to model those relationships.
We interpret synergistic knowledge on simplicial models that are based on semi-simplicial sets. As examples, we investigate the use of consensus objects and the problem of the dining cryptographers. Furthermore, we introduce the axiom system $\syn$ for synergistic knowledge and show that it is sound and complete with respect to the presented simplicial models.
\end{abstract}
\vspace{5pt}
\maketitle  
\section{Introduction}
The successful application of combinatorial topology to distributed systems has sparked the interest of modal logicians.~A topological model interprets the different configurations of a distributed system as an abstract structure called a simplicial complex. Within a simplicial complex, one can express which global states are indistinguishable to processes. If a process cannot tell two configurations apart, then, for any distributed protocol, it must behave the same in both states. Naturally, this approach, exemplified by Herlihy, Kozlov, and Rajsbaum~\cite{DBLP:books/mk/Herlihy2013}, has parallels with formal epistemology. Hence, studying simplicial interpretations of modal logic is a novel and promising subject.

Due to its application to well known problems such as agreement (cf.~Halpern and Moses \cite{DBLP:conf/podc/HalpernM84}), modal logic has been proven to be vital for the formal analysis of distributed Algorithms. Models for modal logic are usually based on a possible worlds approach, where the modal  operator $\square$ is evaluated on Kripke frames. A Kripke frame is a graph whose vertices represent worlds, i.e., global states, and whose edges depict an indistinguishability relation. In a world $w$, a formula $\phi$ is known, denoted by $\square \phi$, if and only if $\phi$ is true in each  world indistinguishable from $w$. These frames can be extended to multi-agent systems by introducing an indistinguishability relation for each agent. A formula $\phi$ is distributed knowledge of a group, first introduced by Halpern and Moses \cite{DBLP:conf/podc/HalpernM84}, if and only if $\phi$ is true in all worlds that cannot be distinguished by any member of the group.

Some recent publications in logic have explored simplicial complexes as models for modal logics with multiple agents. In its simplest form, pure simplicial complexes describe settings in which agents never die and impure complexes are capable of capturing an agent's death. Goubault, Ledent, and Rajsbaum~\cite{DBLP:journals/iandc/GoubaultLR21} pioneered the analysis of simplicial interpretations and laid the foundation for pure complexes as epistemic models. Work regarding impure complexes was conducted independently by van Ditmarsch, Kuznets, and Randrianomentsoa~\cite{DBLP:journals/lmcs/RandrianomentsoaDK23} and by Goubault, Ledent, and Rajsbaum~\cite{DBLP:conf/stacs/GoubaultLR22}. The latter work shows the equivalence between their simplicial models and Kripke models for the  logic $\mathsf{KB4}_n$, whereas van Ditmarsch, Kuznets, and Randrianomentsoa~\cite{DBLP:journals/lmcs/RandrianomentsoaDK23} propose a three-valued logic in which formulas that reference dead agents might be undefined. Van Ditmarsch, Kuzents, and Randrianomentsoa~\cite{DitmarschKuzents23} further compare two-valued with three-valued semantics and categorise simplicial models accordingly. Lastly, Goubault et al.~\cite{goubault2023simplicial} generalise impure complexes based on their containment of sub-worlds, i.e., worlds whose set of alive agents is contained in another indistinguishable world's set of alive agents.

As pointed out by van Ditmarsch et al.~\cite{10.1007/978-3-030-75267-5_1}, improper\footnote{A frame is proper if any two worlds with the same set of alive agents can be distinguished by at least one of them.} Kripke frames cannot be depicted by simplicial complexes. However, relaxing the properties of simplicial complexes might enable us to do so and lead to novel notions of group knowledge. One example of a slightly relaxed version of a simplicial complex is a semi-simplicial set. Simply put, semi-simplicial sets allow us to model global states that consist of the same local states, but differ in their meaning. These kind of scenarios arise in distributed systems when accessing concurrent objects. For example, consider two processes $P$ and $Q$ that concurrently enqueue their values $v_P$ and $v_Q$ to a shared queue. Upon completion, each process locally knows that the queue contains its submitted value. However, they do not know whose value was enqueued first, i.e.,~$P$ and $Q$ cannot distinguish the worlds where the state of the queue is $v_Pv_Q$ or $v_Qv_P$. The complex below in Figure~\ref{fig:intro_queue} illustrates this scenario. 
\begin{figure}[ht]
\begin{center}
\begin{tikzpicture}[scale=0.75]
\node[circle, draw] at (0,0) (w1) [label = below:$ $]{$v_Q$};
\node[circle, draw]at (4,0) (b1)[label = above:$ $]{$v_P$};
\node[circle, fill = white] at (2,1.25) (l1) {$v_Pv_Q$};
\node[circle, fill = white] at (2,-1.3) (l12) {$v_Qv_P$};
\draw[-] (w1) to[out=315, in=225]  (b1);
\draw[-
] (w1) to[out=45, in=135] (b1);
\end{tikzpicture}
\caption{The two processes $P$ and $Q$ cannot distinguish between worlds $v_Pv_Q$ and $v_Qv_P$.}\label{fig:intro_queue}
\end{center}
\end{figure}
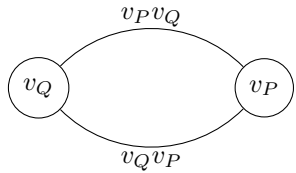
The local state of a process is represented by a vertex that is labelled accordingly and global states correspond to edges. The processes cannot distinguish between the global states because their local states can be part of both. That is, $v_Q$ and $v_P$ belong to both edges. However, we must not restrict the notion of indistinguishability to containing vertices alone. Instead, we can look at higher-order connectivity as well. By doing so,  we introduce a new kind of knowledge that has not yet been explored.

In this paper, we present a semantics for epistemic reasoning on simplicial models that are based on semi-simplicial sets and introduce the synergistic knowledge operator $[G]$. Our notion of synergistic knowledge is capable of describing scenarios in which a group of agents can know more than just the consequences of their pooled knowledge. Different from the higher-order interpretation of distributed knowledge by Goubault et al.~\cite{goubault2023semisimplicial}, which analyses the knowledge of a set of agents, we unfold relations within a group and analyse its knowledge with respect to the internal structure of the group. With this interpretation, the knowledge of two seemingly equal groups can differ due to internal relations between agents. Limiting oneself to sets of agents cannot model this without introducing hybrid agents that represent relations between agents.

Given a finite set of agents $\ag$, we introduce a new modality $[G]$, where the set $G \subseteq \pow(\ag) \setminus \{\emptyset\}$ is called an agent pattern. Agent patterns are best understood representing relations within a group. For example, if $~{G = \{\{a,b,c\}\}}$, then, the agents $a,b,$ and $c$  mutually benefit from each other and $[G]$ reasons about the synergy of the three agents; whereas if $G = \{\{a\},\{b\},\{c\}\}$, the modality $[G]$ describes traditional distributed knowledge. When interpreting the language of synergistic knowledge on simplicial models based on semi-simplicial sets, $[G]\phi$ is true in a simplex $S$ if and only if $\phi$ is true in all simplices that intersect with $S$ in $G$. For example, in Figure \ref{fig:intro_queue}, the agent pattern $\{\{P,Q\}\}$ can distinguish between $v_Pv_Q$ and $v_Qv_P$ because the two edges are different. 

Another contribution of this work is the axiom system $\syn$. We show that $\syn$ is sound and complete with respect to our simplicial models. Our proof employs a variant of the unravelling method used to establish completeness for traditional distributed knowledge (cf. Fagin et al.~\cite{DBLP:books/mit/FHMV1995} and Goubault et al.~\cite{goubault2023semisimplicial}). However, our logic has new axioms that close knowledge under sub-simplices. For example, we cannot only rely on the usual monotonicity axiom that states $[G]\varphi \rightarrow [H]\varphi$, whenever $G\subseteq H$. Instead, we need that
\begin{equation}\label{eq:intro1}
[\{B\}]\phi \rightarrow [G]\phi \quad\text{if there exists  $A \in G$ and $B\subseteq A$.}
\end{equation}
Most notably, \eqref{eq:intro1} can be proved by adding an axiom enforcing that adding agent subpatterns does not create new insights, i.e.,
\[
[G\cup \{B\}]\phi \rightarrow [G]\phi \quad\text{if there exists  $A \in G$ and $B\subseteq A$.}
\]
After introducing the syntax and semantics of $\syn$, we proceed by proving soundness and completeness of $\syn$ with respect to simplicial models. Soundness can be proven right away, but completeness is more involved. To prove completeness, we first introduce different types of Kripke models and show that a weaker version of $\syn$, namely $\synminus$, is sound and complete with respect to those models. After that, we show that our proofs can be extended to Kripke models that are equivalent to simplicial models which are based on semi-simplicial sets. The difference between $\syn$ and $\synminus$ is that $\syn$ characterises proper models.

We first introduce $\kappa$-models, which are similar to the standard $\mathsf{KB}4_n$ Kripke models. It is worth noting that, $\kappa$-models are pseudo-models because they do not satisfy standard group knowledge. That is, given two agent patterns $G$ and $H$, if both cannot distinguish two worlds, neither can their union. However, not satisfying standard group knowledge lets us prove soundness and completeness of $\synminus$ with respect to $\kappa$-models using a canonical model construction. Subsequently, we present $\delta$-models, which are $\kappa$-models that satisfy standard group knowledge. We use unravelling of $\kappa$-models to show that $\synminus$ is sound and complete with respect to $\delta$-models. In a last step, we extend the previous proofs to proper $\delta$-models which shows soundness and completeness of $\syn$ with respect to them. In order to show completeness of $\syn$ with respect to simplicial models, we develop a novel construction that takes a proper $\delta$-frame as input and transforms it into an equivalent complex. This implies completeness of $\syn$ with respect to our simplicial models. The constructive nature of the translation  intuitively demonstrates the relation between the two semantic approaches.

Section \ref{sec:simplicial} presents our notion of a simplicial complex together with an indistinguishability relation for synergistic knowledge. In Section~\ref{sec:logic} we introduce the axiom system $\syn$ and show its soundness with respect to the earlier defined simplicial models. Section \ref{sec:examples} accompanies the presented theory with examples from distributed computing. In Section~\ref{sec: completeness of simplicial models} we prove that $\syn$ is complete with respect to simplicial models. Lastly, in Section~\ref{sec:conclusion} and Section~\ref{sec:future work} we conclude our work and discuss future research directions.

\section{Simplicial structures}\label{sec:simplicial}
In this section, we introduce simplicial structures based on semi-simplicial sets. Let $\ag$ denote the set of finitely many agents and let
\[
\agsi = \{ (A,i) \mid A \subseteq \ag \text{ and } i\in\BN \}.
\]
Let $S\subseteq \agsi$. An element $(A,i)\in S$ is \emph{maximal in $S$} if and only if 
\[
\forall (B,j)\in S. |A|\geq|B|, \text{ where $|X|$ denotes the cardinality of the set $X$.}
\]

\begin{definition}[Simplex]\label{def:simplex}
Let $\emptyset\neq S\subseteq \agsi$. $S$ is a \emph{simplex} if and only if
\begin{enumerate}[$\mathsf{S}$1:]
\item The maximal element is unique, i.e.,
\[
\text{if }(A,i)\in S \text{ and } (B,j)\in S \text{ are maximal in } S \text{, then } A=B \text{ and } i=j.
\]
The maximal element of $S$ is denoted as $\max (S)$. 
\item $S$ is uniquely downwards closed, i.e., for $(B,i)\in S$ and $C\subseteq B$
\[
\exists ! j\in \BN. (C,j)\in S, \text{ where $\exists ! j$ means that there exists exactly one $j$.}
\]
\item $S$ contains nothing else, i.e.,
\[
(B,i)\in S \text{ and } (A,j)=\max(S) \text{ implies } B\subseteq A.
\]
\end{enumerate}
\end{definition}

\begin{definition}[Complex]\label{def:complex}
Let $\BC$ be a set of simplices. $\BC$ is a \emph{complex} if and only if
\begin{enumerate}[$\mathsf{C}$:]
\item For any $S,T  \in \BC$, if there exist $A \subseteq \ag$ and $i \in \mathbb{N}$ with $(A,i) \in S$ and  $(A,i) \in T$, then
\[
\text{for all $B \subseteq A$ and all $j$} \quad (B,j)\in S \quad\text{iff}\quad (B,j)\in T.
\]
\end{enumerate}
\end{definition}

\begin{lemma}\label{lem:maxcomplex}
Let $\BC$ be a complex and $S,T \in \BC$. We find
\[
\max(S)=\max(T) \text{ implies } S=T.
\]
\end{lemma}
\begin{proof}
We show $S \subseteq T$. The other direction is symmetric.
Let $(A,i) = \max(S)$.
Assume $(B,j) \in S$. Because of $\siii$, we have $B \subseteq A$.
By Condition~$\mathsf{C}$, we conclude   $(B,j) \in T$.
\qedhere
\end{proof}

Whenever it is clear from the context, we abbreviate $(\{a_1,...,a_n\},i)$ as $a_1...a_ni$. Furthermore, we omit elements of the form $(\emptyset,i)$ and may use a row (or a mixed row-column) notation to emphasize simplices. For example,
\[
\left\{
\begin{Bmatrix} ab0 \\ a0 \\ b0 \end{Bmatrix},
\begin{Bmatrix} ab1 \\ a0 \\ b1 \end{Bmatrix}
\right\}
\]
is a complex that contains 2 simplices. Whenever we refer to a simplex within a complex, we write $\langle a_1,..,a_ni \rangle$ for the simplex with maximal element $(\{a_1,...,a_n\},i)$. Condition~$\mathsf{C}$ guarantees that this notation is well-defined.

\begin{remark}\label{rem:loc}
Complexes can contain sub-simplices. That is, if $S\in \BC$, then there may exist $T\in \BC$ with $S\subseteq T$. An example of such a complex is
\[
\BC = \left\{
\begin{Bmatrix} abc0\\ ab0, ac0, bc0 \\ a0 , b0, c0 \end{Bmatrix},
\begin{Bmatrix} ab0 \\ a0 , b0 \end{Bmatrix}
\right\}.
\]
\end{remark}

\begin{definition}[Indistinguishability]\label{def:indist}
Let $S\subseteq\agsi$, we define 
\[
S^\circ=\{A \mid \exists i\in \BN: (A,i)\in S\}. 
\]
An agent pattern $G$ is a subset of\/ $\pow(\ag) \setminus \{\emptyset\}$. An agent pattern $G$ cannot distinguish between two simplices $S$ and~$T$, denoted by $S \sim_G T$, if and only if $G \subseteq (S\cap T)^\circ$. 
\end{definition}

We will now prove some properties of $\sim_G$ that motivate it as an indistinguishability relation between simplices.

\begin{lemma}
$\sim_G$ is transitive and symmetric.
\end{lemma}
\begin{proof}
Symmetry immediately follows from the fact that set intersection is commutative.~To show transitivity, let $S,T,U$  be  simplices with $S\sim_G T$ and $T\sim_G U$, i.e.,
\begin{gather}
G  \subseteq (S \cap T)^\circ \label{eq:trans:1}\\
G  \subseteq (T \cap U)^\circ \label{eq:trans:2}
\end{gather}
Let $A \in G$. 
Because of \eqref{eq:trans:1}, there exists $i$ with
\begin{equation} \label{eq:trans:3}
(A,i) \in S
 \quad\text{and}\quad
(A,i) \in T.
\end{equation}
Because of \eqref{eq:trans:2}, there exists $j$ with
\begin{equation} \label{eq:trans:4}
(A,j) \in T
 \quad\text{and}\quad
(A,j) \in U.
\end{equation}
From \eqref{eq:trans:3},  \eqref{eq:trans:4}, and Condition $\sii$ we obtain $i=j$.
Thus by  \eqref{eq:trans:3} and  \eqref{eq:trans:4}, we get $A \in (S \cap U)^\circ $.
Since  $A$ was arbitrary in $G$, we conclude $G\subseteq  (S \cap U)^\circ $.
\qedhere
\end{proof}

\begin{lemma}\label{lem:simp_reflexive}
Let $G$ be an agent pattern and 
\[
G^\star := \{\{a\} \mid \exists A\in G \text{ and } a\in A\}.
\]
Let $\CS_G$ be a maximal set of simplices such that for any $S\in \CS_G$ we have $G^\star \subseteq S^\circ$. The indistinguishability relation $\sim_G$ is reflexive on $\CS_G \times \CS_G$ and empty otherwise.
\end{lemma}
\begin{proof}
We first show reflexivity. 
If $G=\emptyset$, then trivially $G \subseteq (S \cap S)^\circ$ for any $S$.
Assume   $G\neq\emptyset$.
Let $S\in \CS_G$. For each $B \in G$, 
we have to show that $B \in (S \cap S)^\circ$, i.e.,~that 
\begin{equation}\label{eq:sym:3}
\text{there exists $i$ with $(B,i) \in S$}.
\end{equation}
Let $(A,i):=\max(S)$.
Let $b \in B$. Because of  $G^\star \subseteq S^\circ$, 
there exists  $l$ such that $(\{b\},l) \in S$.
By $\siii$ we get $b \in A$. Since $b$ was arbitrary in $B$, we get $B \subseteq A$.
By $\sii$ we conclude that \eqref{eq:sym:3} holds and symmetry is established.

We now show that $\sim_G$ is empty otherwise.
Let $S$ be a simplex such that $~{G^\star \nsubseteq S^\circ}$ and let $T$ be an arbitrary simplex. 
Then there exists $a,A$ with $~{a \in A \in G}$  and $\{a\} \notin S^\circ$, i.e.,
\begin{equation}\label{eq:contra:2}
\text{for all $i$, $(\{a\},i) \notin S$.}
\end{equation}
Suppose towards a contradiction that 
\begin{equation}\label{eq:contra:1}
G \subseteq (S \cap T)^\circ
\end{equation}
Because of $A \in G$ we get $A \in  (S \cap T)^\circ$.
Hence $A \in S^\circ$, i.e., there exists $l$ with $(A,l) \in S$.
With S2 and $\{a\} \subseteq A$ we find that there exists $j$ with $(\{a\},j)\in S$. This is a contradiction to \eqref{eq:contra:2}.
Thus \eqref{eq:contra:1} cannot hold.
\qedhere
\end{proof}

\begin{corollary}\label{cor:equiv_rel}
$\sim_G$ is an equivalence relation on $\CS_G\times \CS_G$.
\end{corollary}

\begin{lemma}\label{l:extendG:1}
Let $\BC$ be a complex and 
$S, T \in \BC$. Further, let $A \in (S \cap T)^\circ$ and $B \subseteq A$. 
We find $B \in (S \cap T)^\circ$.
\end{lemma}
\begin{proof}
From $A \in (S \cap T)^\circ$, we obtain that there exists $i$ such that $(A,i) \in S$ and $(A,i) \in T$.
From $\sii$ we find that there exists $j$ such that $(B,j) \in S$.
Thus, by $\mathsf{C}$, we get  $(B,j) \in T$ and we conclude $B \in (S \cap T)^\circ$.
\end{proof}

\begin{corollary}\label{cor:equiv}
Let $G$ be an agent pattern and let $A,B \subseteq \ag$ such that $~{B\subseteq A \in G}$. It holds that
\[
 \sim_{G\cup \{B\}}     \ =\  \sim_{G}.
\]
\end{corollary}

\begin{lemma}\label{lem:std_knowledge}
Let $G$ be an agent pattern. It holds that
\begin{equation}\label{eq:D}
 \bigcap\limits_{B\in G}\sim_{\{B\}} \ = \ \sim_G. \tag{$\mathsf{D}$}
\end{equation}
\end{lemma}
\begin{proof}
$(S,T) \in \bigcap_{B\in G}\sim_{\{B\}}$ if{f}
for each  $B\in G$, we have  $B \in (S \cap T)^\circ$ if{f} 
$G\subseteq(S\cap T)^\circ$   if{f} $S \sim_G T$.
\qedhere
\end{proof}

\begin{remark}[Standard group knowledge]
The property \eqref{eq:D} is also referred to as standard group knowledge.
\end{remark}

\begin{lemma}[Anti-Monotonicity]\label{l:anti:_1}
$G \subseteq H$ implies $\sim_H \subseteq \sim_G$.
\end{lemma}
\begin{proof}
Assume $G \subseteq H$. For any two simplices $S$ and $T$ with $S\sim_H T$, we have $G \subseteq H \subseteq (S\cap T)^\circ$ by Definition \ref{def:indist} and hence $S \sim_G T$. \qedhere
\end{proof}

The next lemma states that adding synergy to an agent pattern makes it stronger in the sense that it can distinguish more simplices. 

\begin{lemma}\label{lem:clo}
Let $H_1, H_2,\ldots, H_n \subseteq \ag$ with $n\geq2$
We have 
\[
 \sim_{\{H_1 \cup H_2, \ldots, H_n\}}     \ \subseteq\  \sim_{\{H_1, H_2, \ldots, H_n\}}.
\]
\end{lemma}
\begin{proof}
From Lemma~\ref{l:extendG:1} and Lemma~\ref{l:anti:_1}  we find that
\[
\sim_{\{H_1 \cup H_2, \ldots, H_n\}}   \ =\  \sim_{\{H_1 \cup H_2, H_1, H_2, \ldots, H_n\}}  \ \subseteq \  \sim_{\{H_1, H_2, \ldots, H_n\}}.  
\]
\qedhere
\end{proof}

In traditional Kripke semantics, distributed knowledge of a set of agents is modelled by considering the indistinguishability relation that is given by the intersection of the indistinguishability relations of the individual agents.
The following lemma states that in our framework, this intersection corresponds to the agent pattern consisting of singleton sets for each agent.

\begin{lemma}\label{lem:dist_knowledge}
Let $G \subseteq \ag$ and $H = \bigcup_{a\in G}\{\{a\}\}$. We have 
\[
 \bigcap\limits_{a\in G}\sim_{\{\{a\}\}} \ = \ \sim_H.
\]
\end{lemma}
\begin{proof}
$(S,T) \in \bigcap_{a\in G}\sim_{\{\{a\}\}}$ if{f}
for each  $a\in G$, we have  $\{a\} \in (S \cap T)^\circ$ if{f} (by the definition of $H$)
$H\subseteq(S\cap T)^\circ$   if{f} $S \sim_H T$.
\end{proof}

\section{Logic}\label{sec:logic}
Let $\P$ be a countable set of atomic propositions. The logic of synergistic knowledge is a normal modal logic that includes a modality $[G]$ for each agent pattern $G$.  
Formulas of the language of synergistic knowledge $\CL$ are inductively defined by the following grammar:
\[
\phi ::= p \mid  \lnot  \phi    \mid \phi \land \phi \mid [G] \phi
\]
where $p \in \P$ and $G$ is an agent pattern. The remaining Boolean connectives are defined as usual. In particular, we set $\bot:= p \land \lnot p$ for some fixed $p \in \P$, and we write $\lalive(G)$ for $\neg[G]\bot$. If $G$ is an agent pattern, then $G^C$ denotes its complement. That is, 
\[
G^C := \{H \in \pow(\ag)\setminus \{\emptyset\} \mid  \nexists B\in G. H \subseteq B\}.
\]
Moreover, we define
\[
\dead(G) := \bigwedge_{B \in G}\neg \lalive(\{B\}).
\]
Notice that $\dead(G) \not \equiv \neg \lalive(G)$. Indeed, $\dead(G)$ expresses that for each $B\in G$, the pattern $\{B\}$ is dead, whereas $\neg\lalive(G)$ is true if some $\{B\}\subseteq G$ is dead.
The axiom system $\syn$ consists of the axioms:

\begin{gather}
\text{all propositional tautologies}\tag{Taut}
\label{eq:taut:1}\\
[G](\phi \rightarrow \psi) \rightarrow([G]\phi \rightarrow [G]\psi)\tag{K}
\label{eq:taut:1}\\
\phi \rightarrow [G]\neg[G]\neg\phi \tag{B}
\label{eq:B:1}\\
[G]\phi \to [G][G]\phi \tag{4}
\label{eq:4:1}\\
\lalive(G) \to ([G]\phi \to \phi) \tag{T}
\label{eq:T:1}\\
\lalive(G) \land \dead(G^C) \land \phi \rightarrow [G]( \dead(G^C)\rightarrow \phi) \tag{P}
\label{eq:P}\\
\bigvee_{G \subseteq \pow(\ag) \setminus \{\emptyset\}}\lalive(G) \tag{NE}
\label{eq:NE}\\
[G]\phi \rightarrow [H]\phi \quad\text{if $G\subseteq H$}\tag{Mono}
\label{eq:mono:1}\\
[G\cup \{B\}]\phi \rightarrow [G]\phi \quad\text{if there exists  $A \in G$ and $B\subseteq A$} \tag{Equiv}
\label{eq:union:1}\\
\lalive(G) \land \lalive (H) \to \lalive (G \cup H )\tag{Union}
\label{eq:union:1}\\
\lalive(G) \to \lalive(\{A \cup B\})  \quad\text{if $A,B \in G$} \tag{Clo}\label{eq:clo:1}
\end{gather}
and the inference rules modus ponens (MP) and $[G]$-necessitation ($[G]$-Nec). We write $\vdash \varphi$ to denote that $\varphi \in \CL$ can be deduced in the system $\syn$.
\[
  \begin{prooftree} 
    \hypo{A}
    \hypo {A\rightarrow B}
    \infer2{B}
  \end{prooftree}  \qquad(\text{MP})\qquad
  \begin{prooftree}
    \hypo{A}
    \infer1{[G]A}
  \end{prooftree}\qquad(\text{$[G]$-Nec})
\]
Notice that 
\begin{equation}\label{eq:loc}
\lalive(G) \land \dead(G^C) \land \phi \rightarrow [G]\phi,
\end{equation}
is not derivable in $\syn$, because, as mentioned in Remark~\ref{rem:loc}, complexes can contain sub-worlds. In order to derive Equation \eqref{eq:loc} we would have to add 
\[
\lalive(G) \land \dead(G^C) \rightarrow [G]\dead(G^C)
\]
to $\syn$ (cf. Goubault et al.~\cite{goubault2023simplicial}).
However, if $~{\ag \in G}$, we have $G^C = \emptyset$, i.e.,
\[
\lalive(G) \land \phi \rightarrow [G]\phi
\]
because the empty conjunction evaluates to $\top$.

\begin{definition}[Simplicial model]
A simplicial model $\CC = (\BC,L)$ is a pair such that
\begin{enumerate}
\item $\BC$ is a complex;
\item $L:\BC\rightarrow \pow(\P)$ is a valuation.
\end{enumerate}
\end{definition}

\begin{definition}[Truth]\label{def:truth}
Let $\CC = (\BC,L)$ be a simplicial model, $S \in \BC$, and $\phi \in \CL$.
We define $\CC,S \Vdash_{\sigma} \phi$ inductively by

\begin{align*}
  & \CC,S \Vdash_{\sigma}  p		&\text{if{f}}\qquad & p \in L(S)\\
    &\CC,S \Vdash_{\sigma}  \neg \phi 			&\text{if{f}}\qquad 	& \CC,S  \not \Vdash_{\sigma} \phi \\
  & \CC,S \Vdash_{\sigma} \phi \land \psi  	&\text{if{f}}\qquad 	&\CC,S \Vdash_{\sigma} \phi \text{ and } \CC,S \Vdash_{\sigma} \psi\\
      &\CC,S \Vdash_{\sigma}  [G] \phi \qquad &\text{if{f}} \qquad &S \sim_G T \text{ implies } \CC,T \Vdash_{\sigma} \phi \quad\text{for all $T \in \BC$}.
\end{align*}
\end{definition}

We write $\CC \Vdash_{\sigma} \phi$, if  $\CC,w \Vdash_{\sigma} \phi$ for all $S \in \BC$.
A formula $\phi$ is $\sigma$-valid, denoted by $ \Vdash_{\sigma} \phi$, if $\CC \Vdash_{\sigma}\phi$ for all models $\CC$. Whenever it is clear from the context we omit the subscript $\sigma$.

Corollary~\ref{cor:alive} is an immediate consequence of previously introduced Corollary~\ref{cor:equiv_rel}. It relates the formula $\lalive(G)$ to the structure of the underlying complex as expected.

\begin{corollary}\label{cor:alive}
Let $\CC$ be a simplicial model. For any agent pattern $G$, we find
\[
\CC,S \Vdash_\sigma \lalive (G) \quad\text{iff}\quad S\sim_G S.
\]
\end{corollary}

Soundness of $\syn$ with respect to simplicial models follows as usual. We present the proof of completeness in Section \ref{sec: completeness of simplicial models}.

\begin{theorem}[Soundness]\label{thm:simp_soundness}
$\vdash \varphi$ implies $\Vdash_\sigma \varphi$.
\end{theorem}
\begin{proof}
We only show (T), (P), (NE), (Union), (Clo), (Mono), (Equiv), and ($[G]$-Nec). Let $\CC = (\BC,L)$ be an arbitrary model.
\begin{enumerate}
\item (T): Consider a world $S\in \BC$ and assume that $~{\CC, S \Vdash_\sigma \lalive(G)}$ and $~{\CC,S \Vdash_\sigma [G]\varphi}$. By Corollary \ref{cor:alive}, we find $S \sim_G S$ and thus $\CC,S \Vdash_\sigma \varphi$ because $\CC,S \Vdash_\sigma [G]\varphi$.
\item (P): Consider a world $S\in \BC$ and assume $~{\CC, S\Vdash \lalive(G) \land \dead(G^C) \land \phi}$ for some $\phi \in \CL$. By assumption, $G$ must have an unique maximal element. Indeed, towards a contradiction, assume that it is not the case and $G$ has the maximal elements $B_1,B_2,...,B_n$. Consider now 
\[
B = \bigcup^n_{i=1} B_i.
\]
It is straightforward to verify that $\CC, S \Vdash \lalive(\{B\})$. However, since $~{B\in G^C}$, by assumption, it holds that $\CC,S \Vdash \dead(\{B\})$ which is a contradiction. Additionally, by $\siii$, the maximal element of $G$, say $A$, is the set of all agents alive in $S$. Let $T\in \BC$ be such that $S\sim_G T$ and assume $~{\CC, T \Vdash \dead(G^C)}$. By transitivity of $\sim_G$ and Corollary~\ref{cor:alive}, we have that $~{\CC, T \Vdash \lalive(G)}$. Thus, by the same reasoning as before, $A$ is the set of all alive agents for $T$ as well, i.e., $\max(S) = (A,i)$ and $\max(T)= (A,i)$ for some $i\in \mathbb{N}$. Finally, by Lemma~\ref{lem:maxcomplex}, we find that $S=T$, and thus $\CC,T \Vdash \phi$.
\item (NE): Follows because simplices are not empty.
\item (Union):  Consider a world $S\in \BC$ and assume that $~{\CC, S \Vdash_\sigma \lalive(G)}$ and $~{\CC, S \Vdash_\sigma \lalive(H)}$. By Corollary \ref{cor:alive} it holds that $S\sim_G S$ and $~{S\sim_H S}$. By Lemma \ref{lem:std_knowledge} we have 
\[
\sim_{G\cup H} = \bigcap_{B \in G\cup H} \sim_{\{B\}} = \left( \bigcap_{B\in G}\sim_{\{B\}} \right) \cap \left( \bigcap_{B\in H}\sim_{\{B\}} \right),
\]
and thus $S \sim_{G\cup H} S$ and $~{\CC,S\Vdash_\sigma \lalive(G\cup H)}$.
\item (Clo): Consider a world $S\in \BC$ and assume that $~{\CC, S \Vdash_\sigma \lalive(G)}$, and let $A,B\in G$. By Lemma \ref{l:extendG:1} we find $S \sim_{\{A\}} S$ and $S\sim_{\{B\}} S$, i.e., there exist $i,j\in \mathbb{N}$ such that $(A,i) \in S$ and $(B,j)\in S$. Furthermore, let $(C,k) = \max(S)$. By $\siii$, we find $A \subseteq C$ as well as $B\subseteq C$, and thus $A \cup B \subseteq C$. Since $S$ is downwards closed by $\sii$, there exists $k\in \mathbb{N}$ such that $(\{A\cup B\},k)\in S$. Hence, $S \sim_{\{A \cup B\}}S$ and $\CC,S \Vdash_\sigma \lalive(\{A \cup B\})$ by Lemma \ref{lem:alive}.
\item (Mono): Follows from Lemma \ref{l:anti:_1}.
\item (Equiv): Follows from Corollary \ref{cor:equiv}.
\end{enumerate}
Lastly, we show ($[G]$-Nec). Let $A \in \mathcal{L}$ and assume that $A$ is $\sigma$-valid, i.e., for any simplicial model $\CC = (\BC,L)$ and $S\in \BC$, we have $\CC, S \Vdash_\sigma A$. Let $S\in \BC$ be arbitrary. By assumption, it holds that for all $T$ with $S \sim_G T$, we have $~{\CC,T \Vdash_\sigma A}$. Thus, $~{\CC,S \Vdash_\sigma [G]A}$ by the definition of truth. Since $S$ was arbitrary, we have $~{\CC \Vdash_\sigma [G]A }$. Lastly, since $\CC$ was arbitrary, $[G]A$ is $\sigma$-valid, i.e., $~{\Vdash_\sigma [G]A}$.
\end{proof}

\section{Examples}\label{sec:examples}
This section illustrates the application of our logic to distributed systems with the help of two examples. We interpret synergy as having access to shared primitives in both examples. In Example \ref{example:consensus number}, the shared primitive is a consensus object and in Example \ref{example:dining cryptographers} it is a shared coin. Given three agents, Example \ref{example:consensus number} captures the idea that for some applications, the agent pattern must include the area of the triangle and not just its edges. Thus, Example \ref{example:consensus number} shows the difference between mutual and pairwise synergy. A similar example can be found in the extended abstract of this work~\cite{DBLP:conf/clar/CachinLS23}. Example \ref{example:dining cryptographers} demonstrates that the patterns $\{\{a,b\},\{a,c\}\}$, $\{\{a,b\},\{b,c\}\}$, and $\{\{b,c\},\{a,c\}\}$ are weaker than the pattern $\{\{a,b\},\{a,c\}, \{b,c\}\}$. In other words, pairwise synergy between all agents is stronger than pairwise synergy among some agents.

Regarding notation, from now on we will omit the set parentheses for agent patterns whenever it is clear from the context and write for example $[abc,ab,ac]$ instead of $[\{\{a,b,c\},\{a,b\},\{a,c\}\}]$.

\begin{example}[Consensus number]\label{example:consensus number}
A $n$-consensus protocol is implemented by $n$ processes that communicate through shared objects. The processes each start with an input of either $1$ or $0$ and must decide a common value. A consensus protocol must ensure that 

\begin{enumerate}
\item Consistency: all processes must decide on the same value.
\item Wait-freedom: each process must decide after a finite number of steps.
\item Validity: the common decided value was proposed by some process.
\end{enumerate}

Herlihy \cite{DBLP:journals/toplas/Herlihy91} defines the consensus number of an object $O$ as the largest $n$ for which there is a consensus protocol for $n$ processes that only uses finitely many instances of $O$ and any number of atomic registers. It follows from the definition that  no combination of objects with a consensus number of $k<n$ can implement an object with a consensus number of $n$. 

We can represent the executions of a $n$-consensus protocol as a tree in which one process moves at a time. By validity and wait-freedom, the initial state of the protocol must be bivalent (i.e.,~it is possible that $0$ or $1$ are decided), and there must exist a state from which on all successor states are univalent. Hence, the process that moves first in such a state decides the outcome of the protocol. Such a state is called a critical state.

In order to show that an object has a consensus number strictly lower than $k$, we derive a contradiction by assuming that there is a valid implementation of a $k$-consensus protocol. Next, we maneuver the protocol into a critical state and show that the processes will not be able to determine which process moved first. Therefore, for some process $P$, there exist two indistinguishable executions in which $P$ decides differently. However, if the object has a consensus number of $k$, the processes will be able to tell who moved first.

Synergetic knowledge is able to describe the situation from the critical state onwards. We interpret an element $\{p_1,...,p_k\}$ of a synergy pattern $G$ as the processes $p_1$ up to $p_k$ having access to objects with a consensus number of $k$. For each process $p_i$, we define a propositional variable $\mathsf{move}_i$ that is true if $p_i$ moved first at the critical state. Furthermore, we define
\[
\varphi_i := \mathsf{move}_i \land \bigwedge_{ 1\leq j \leq n \text{ and } j \neq i}\neg \mathsf{move}_j,
\]
i.e., if $\varphi_i$ is true, then the $i$-th process moved first.
Let $\CC = (\BC,L)$ be a model, if $\CC \Vdash_\sigma [G]\varphi_1 \lor [G]\varphi_2 \lor \dots \lor [G]\varphi_n$ holds in the model, then it is always possible for the processes in $G$ to tell who moved first. Lastly, if $G$ has $n$ agents, we have for any $G'$ with less than $n$ agents
\[
\CC \not \Vdash_\sigma [G']\varphi_1 \lor [G']\varphi_2 \lor \dots \lor [G']\varphi_n,
\]
which means that the access to objects with a consensus number of $n$ is required. 

For three agents $a,b$ and $c$, the model $\CC = (\BC,L)$ is given by 

\[\BC = \left\{
\begin{Bmatrix} abc0 \\ ab0 \\ bc0 \\ ac0 \\ a0,b0,c0 \end{Bmatrix},
\begin{Bmatrix} abc1 \\ ab0 \\ bc0 \\ac0 \\ a0,b0,c0  \end{Bmatrix},
\begin{Bmatrix} abc2 \\ ab0 \\ bc0 \\ ac0 \\ a0,b0,c0  \end{Bmatrix}
\right\}
\]
with a valuation $L$ representing that someone moved first, i.e.,
\[
\CC, \langle abc0 \rangle \Vdash_\sigma \varphi_a \qquad \CC, \langle abc1 \rangle \Vdash_\sigma \varphi_b \qquad \CC, \langle abc2 \rangle \Vdash_\sigma \varphi_c.
\]
It is easy to check that $\langle abc0 \rangle \sim_{ab,ac,bc} \langle abc1 \rangle $ and hence, having access to an object with consensus number 2 is not enough in order to distinguish those worlds. However, 
\[
\CC \Vdash_\sigma [abc]\varphi_a \lor [abc]\varphi_b \lor [abc]\varphi_c
\] 
is true and shows that access to objects with consensus number 3 suffices.
\end{example}

\begin{figure}[ht]
\begin{center}
\begin{tikzpicture}[scale=0.65]

\node[circle, draw] at (3,0) (c) [label = below:$ $]{$c0$};
\node[circle, draw] at (0,5.196) (b)[label = above:$ $]{$b0$};
\node[circle, draw] at (-3,0) (a) [label = left:$ $]{$a0$};

\draw[-] (c) to node[midway, below, rotate = -60]{$bc0$} (b);
\draw[-] (c) to node[midway, below, rotate = 0]{$ac0$} (a);
\draw[-] (b) to node[midway, below, rotate = 60]{$ab0$} (a);

\draw[-] (a) to[out=315, in=225]  node[midway, below, rotate = 0]{$ac1$} (c);
\draw[-] (a) to[out=115, in=180]  node[midway, below, rotate = 60]{$ab1$}(b);
\draw[-] (c) to[out=60, in=0]   node[midway, below, rotate = -60]{$bc1$}(b);

\end{tikzpicture}
\caption{Dining cryptographers model.}\label{fig:dining_cryptographers}
\end{center}
\end{figure}
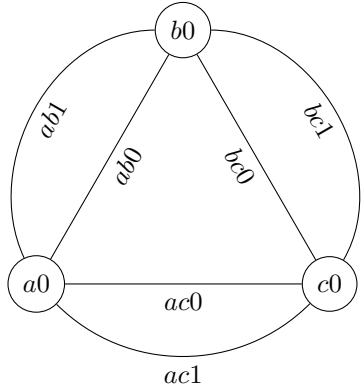

\begin{example}[Dining cryptographers]\label{example:dining cryptographers}
The dining cryptographers problem, proposed by Chaum \cite{DBLP:journals/joc/Chaum88}, illustrates how a shared-coin primitive can be used by three cryptographers (i.e.,~agents) to find out whether their employer or one of their peers paid for the dinner. However, if their employer did not pay, the payer wishes to remain anonymous. 

For the sake of space, we do not give a full formalisation of the dining cryptographers problem.~Instead, we solely focus on the ability of agreeing on a coin-flip and the resulting knowledge.~In what follows, we will provide a model in which the agents $a,b$ and $c$ can determine whether or not their employer paid if and only if they have pairwise access to a shared coin.

Let the propositional variable $p$ denote that their employer paid.
We interpret an agent pattern  $G = \{\{a,b\}\}$ as $a$ and $b$, having access to a shared coin. Our model $\CC = (\BC,L)$, depicted in Figure \ref{fig:dining_cryptographers}, is given by the complex 

\begin{align*}
\BC = \begin{Bmatrix} &\begin{Bmatrix} abc0 \\ ab0 \\ bc0 \\ ac0 \\ a0,b0,c0 \end{Bmatrix},
\begin{Bmatrix} abc1 \\ ab1 \\ bc0 \\ ac0 \\ a0,b0,c0  \end{Bmatrix},
\begin{Bmatrix} abc2 \\ ab0 \\ bc1 \\ ac0 \\ a0,b0,c0  \end{Bmatrix},
\begin{Bmatrix} abc3 \\ ab0 \\ bc0 \\ ac1 \\ a0,b0,c0 \end{Bmatrix},\\
\\
&\begin{Bmatrix} abc4 \\ ab1 \\ bc1 \\ ac0 \\ a0,b0,c0  \end{Bmatrix},
\begin{Bmatrix} abc5 \\ ab1 \\ bc0 \\ ac1 \\ a0,b0,c0  \end{Bmatrix},
\begin{Bmatrix} abc6 \\ ab0 \\ bc1 \\ ac1 \\ a0,b0,c0 \end{Bmatrix},
\begin{Bmatrix} abc7 \\ ab1 \\ bc1 \\ ac1 \\ a0,b0,c0  \end{Bmatrix} \end{Bmatrix}
\end{align*}
and the valuation $L$ is chosen such that
\begin{align*}
&p \in L(\langle abc0 \rangle),
&p \not\in L(\langle abc1 \rangle), & &
&p \not\in L(\langle abc2 \rangle),
&p \not\in L(\langle abc3 \rangle),\\
&p \in L(\langle abc4 \rangle),
&p \in L(\langle abc5 \rangle),& &
&p \in L(\langle abc6 \rangle),
&p \not\in L(\langle abc7 \rangle).
\end{align*}
Consider the agent pattern $G = \{\{a,b\},\{a,c\}, \{b,c\}\}$, then 
\begin{equation}
\CC  \Vdash_\sigma [G]p \lor [G]\neg p,
\end{equation}
i.e.,~in any world, if all agents have pairwise access to shared coins, they can know the value of $p$. Furthermore, for each $H \subsetneq G$ and each $S\in \BC$ 
\begin{equation}\label{eq:dining_cryptographers}
\CC,S \not \Vdash_\sigma [H]p \lor [H]\neg p.
\end{equation}
Notice that (\ref{eq:dining_cryptographers}) states that there is no world, where an agent pattern $H$ can know whether $p$ or $\neg p$, and hence, it is stronger than $\CC \not \Vdash_\sigma [H]p \lor [H]\neg p$.
\end{example}

\section{Completeness of $\syn$}\label{sec: completeness of simplicial models}
In order to show completeness of $\syn$ with respect to simplicial models, we take a detour via Kripke-models and focus on $\syn$ without the axiom (P), which is hereafter denoted by $\synminus$. We start by introducing $\kappa$-models which represent common multi-agent models  in which \eqref{eq:D} is not necessarily satisfied. Thus, we can employ the standard techniques to show that $\synminus$ is sound and complete with respect to $\kappa$-models. Next, we present $\delta$-models, which are $\kappa$-models that satisfy \eqref{eq:D}. The proof that $\synminus$ is complete with respect to $\delta$-models is more challenging and requires the so-called unravelling method. Once we showed soundness and completeness of $\synminus$ with respect to those models, we show that $\syn$ is sound and complete with respect to their proper\footnote{Properness will be defined later.} versions. Lastly, we proceed with proving completeness with respect to simplicial models by relating them to proper $\delta$-models.

\subsection{$\kappa$-models}\label{subsec:kmodel}
In this section, we introduce $\kappa$-models and show that $\synminus$ is sound and complete with respect to $\kappa$-models. Definition \ref{def:pre-model} introduces pre-models, which are our most simple models. All subsequent models are pre-models.

\begin{definition}[Pre-model]\label{def:pre-model}
A pre-model $\CM = (W,\sim,V)$ is a tuple where 
\begin{enumerate}
\item $W$ is a set of possible worlds;
\item $\sim$ is a function that assigns to each agent pattern $G$ a symmetric and transitive relation $\sim_G$ on $W$;
\item $V:W \rightarrow \pow(\P)$ is a valuation.
\end{enumerate}
\end{definition}

\begin{remark}[Notation]
In order to avoid confusion due to the overloading of $\sim_G$, we use lower-case letters for worlds of a pre-model, i.e., $w \sim_G v$, and capital letters for simplices, i.e.,~$S\sim_G T$.
\end{remark}

\begin{definition}[Truth]\label{def:truth}
Let $\CM = (W,\sim,V)$ be a pre-model, $w \in W$, and$~{\phi \in \CL}$.
We define $\CM,w \Vdash \phi$ inductively by
\begin{align*}
  & \CM,w \Vdash  p		&\text{if{f}}\qquad & p \in V(w)\\
    &\CM,w \Vdash  \neg \phi 			&\text{if{f}}\qquad 	& \CM,w  \not \Vdash \phi \\
  & \CM,w \Vdash \phi \land \psi  	&\text{if{f}}\qquad 	&\CM,w \Vdash \phi \text{ and } \CM,w \Vdash \psi\\
      &\CM,w \Vdash  [G] \phi \qquad &\text{if{f}} \qquad &w \sim_G v \text{ implies } \CM,v \Vdash \phi \quad\text{for all $v \in W$}.      
\end{align*}
\end{definition}
We write $\CM \Vdash \phi$, if  $\CM,w \Vdash \phi$ for all $w \in W$.

\begin{definition}\label{def:alive}
Let  $\CM = (W,\sim, V)$ be a pre-model, we define
\[
\alive(G)_\CM = \{ w \mid {w\sim_G w} \}.
\] 
\end{definition}

If the pre-model is clear from the context, we omit the subscript $\CM$ and write $\alive(G)$ instead of $\alive(G)_\CM$.

\begin{lemma}\label{lem:sound:alive}
Let $\CM = (W,\sim,V)$ be a pre-model. It holds that
\[
\CM,w \Vdash \lalive(G) \quad\text{iff}\quad w\in\alive(G)_\CM.
\]
\end{lemma}
\begin{proof}

We first show that
\[
\CM,w \Vdash \lalive(G) \quad\text{implies}\quad w\in \alive(G)_\CM.
\]
Assume $\CM,w \Vdash \lalive(G)$, i.e.,~$\CM,w \not \Vdash [G]\bot $. By the definition of truth, there must exist $v\in W$ with $w\sim_Gv$ and $\CM,v \Vdash \top$. By symmetry,  we have $v\sim_G w$ and by transitivity we have $w\sim_Gw$. Hence, $w\in \alive(G)_\CM$. We now show~that

\[
w\in\alive(G)_\CM \quad\text{implies}\quad \CM,w \Vdash \lalive(G).
\]
Assume $w\in\alive(G)_\CM$. By definition of $\alive(G)_\CM$, we have that $w\sim_Gw$. Therefore, $\CM,w \not \Vdash [G]\bot$ by the definition of truth.
\end{proof}

\begin{definition}[$\kappa$-model]\label{def:k model}
Let  $\CM = (W,\sim, V)$ be a pre-model. $\CM$ is called a $\kappa$-model if and only if, for all agent patterns $G$ and $H$:
\begin{enumerate}[$\mathsf{K}1$:]
\item $\alive(G)_\CM  \cap \alive(H)_\CM  \subseteq \alive(G\cup H)_\CM $;
\item $\alive(G)_\CM  \subseteq \alive(\{A\cup B\})_\CM $ for $A,B\in G$;
\item $\sim_H \subseteq \sim_G$, if $G\subseteq H$;
\item $ \sim_G \subseteq \sim_{G\cup \{B\}} $ if there exists $A\in G$ with $B\subseteq A$;
\itemne for all $w\in W$, there exists an agent pattern $G$ such that $w\sim_G w$.
\end{enumerate}
\end{definition}
A formula $\phi$ is $\kappa$-valid, denoted by $\Vdash_\kappa \phi$, if $\CM \Vdash \phi$ for all $\kappa$-models $\CM$.
\begin{lemma}\label{lem:k4}
Let $\CM = (W,\sim,V)$ be a $\kappa$-model. Furthermore, let $G$ be an agent pattern such that there exists $A\in G$ with $B\subseteq A$. It holds that  $\sim_G \subseteq \sim_{\{B\}}$. 
\end{lemma}
\begin{proof}
By $\kthree$ and $\kfour$ we have $\sim_G \subseteq \sim_{G\cup \{B\}} \subseteq \sim_{\{B\}}$.\qedhere
\end{proof}

\begin{remark}
Let $\CM = (W,\sim,V)$ be a $\kappa$-model and let $G$ be an agent pattern as in Lemma \ref{lem:k4}. By $\kthree$, we have $\sim_G = \sim_{G \cup \{B\}}$.
\end{remark}

The property $\kone$ ensures that for each world, there exists a maximal alive agent pattern and $\ktwo$ forces $\lalive(G)$ to be downwards closed. $\kthree$ guarantees that an agent pattern $G$ cannot know more than its supersets, and $\kfour$ states that adding subpatterns to $G$ does not strengthen its knowledge. Moreover, $\mathsf{NE}$ ensures that there are no empty-worlds, i.e., worlds in which no agent is alive. However as shown  below in Example \ref{example:no std dk}, $\kappa$-models do not necessarily satisfy \eqref{eq:D}, i.e., 
\begin{equation*}
 \bigcap\limits_{B\in G}\sim_{\{B\}} \ = \ \sim_G. 
\end{equation*}

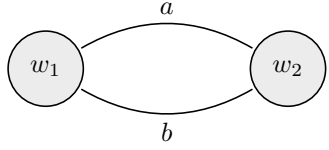
\begin{figure}[ht]
\begin{center}
\begin{tikzpicture}[modal]
  \node[world] (w) {$w_1$};
  \node[world] (v) [right=2.2cm of w] {$w_2$};
  \path[-] (w) edge[bend left] node[above] {$a$} (v);
   \path[-] (w) edge[bend right] node[below] {$b$} (v);
\end{tikzpicture}
\end{center}
\caption{The $\kappa$-model used in Example \ref{example:no std dk} does not satisfy standard group knowledge. Reflexive arrows are implicit.}\label{fig:kappa_std}
\end{figure}

\begin{example}\label{example:no std dk}
Let $\ag = \{a,b\}$ and consider the $\kappa$-model $\CM$ shown in Figure~\ref{fig:kappa_std}, where reflexive arrows are implicit. The agents $a$ and $b$ cannot distinguish between the worlds $w_1$ and $w_2$, i.e., $(w_1,w_2)\in \sim_{\{a\}}$ and $(w_1,w_2)\in \sim_{\{b\}}$. However, together they are able to tell the two worlds apart. That is, $(w_1,w_2) \not\in \sim_G$ for any agent pattern $G$ that contains the set $\{\{a\},\{b\}\}$. Hence, $\CM$ does not satisfy standard group knowledge.
\end{example}
Lastly, notice, that $\kthree$ and $\kfour$ affect the formula $\lalive(G)$ as well. For example, the following formula is $\kappa$-valid
\[
\lalive(G) \rightarrow \lalive(\{B\}) \quad \text{ if there exists $A\in G$ and $B \subseteq A$.}
\]
Indeed, let $\CM = (W,\sim,V)$ be a $\kappa$-model and assume $\CM,w \Vdash \lalive(G)$. By definition, there exists $v\in W$ with $w\sim_G v$ and $\CM,v \Vdash \top$. By symmetry and transitivity of $\sim_G$ we obtain $w\sim_G w$. By Lemma $\ref{lem:k4}$ we have that $(w,w)\in \sim_{\{B\}}$ and thus $\CM,w \Vdash \lalive(\{B\})$. 

\begin{theorem}[Soundness]\label{Thm:sound_kappa}
$\synminus$ is sound with respect to $\kappa$-models.
\end{theorem}
\begin{proof}
We only show (T), (NE), (Union), (Clo), (Mono), (Equiv) and ($[G]$-Nec).

\begin{enumerate}
\item (T): Consider a world $w\in W$ and assume that $\CM,w\Vdash \lalive(G)$ and $\CM,w\Vdash [G]\phi$. By Lemma \ref{lem:sound:alive} we have $w\in \alive(G)$, i.e., $w\sim_Gw$. By the definition of truth we have $\CM,w \Vdash \phi$.

\item (NE): Let $w\in W$ be arbitrary. By $\mathsf{NE}$, there exists an agent pattern $G$ such that $w\sim_G w$. By Lemma \ref{lem:sound:alive}, $\CM, w \Vdash \lalive(G)$ and thus 
\[
\CM, w \Vdash \bigvee_{G \subseteq \pow(\ag)\setminus \{\emptyset\}}\lalive(G).
\]

\item (Union): Assume $\CM,w \Vdash \lalive(G)$ and $\CM,w \Vdash \lalive(H)$. By Lemma~\ref{lem:sound:alive}, $w\in \alive(G)\cap\alive(H)$. By $\kone$, $w\in \alive(G\cup H)$ and by Lemma~\ref{lem:sound:alive} $\CM,w \Vdash \lalive(G \cup H)$.

\item (Clo): Assume $\CM,w \Vdash \lalive(G)$ and let $A,B \in G$. By Lemma \ref{lem:sound:alive}, we have $w\in \alive(G)$. By $\ktwo$, $w\in \alive(\{A\cup B\})$ and by Lemma \ref{lem:sound:alive} we obtain $\CM,w \Vdash \lalive(\{A\cup B\})$.

\item (Mono): Assume that $G\subseteq H$ for arbitrary $G,H$. Let $w\in W$ be arbitrary such that $\CM,w \Vdash [G]\phi$. By the definition of truth, $w\sim_G v$ implies $\CM,v \Vdash \phi$ for all $v\in W$. By $\kthree$ we have that $\sim_H \subseteq \sim_G$, i.e., $w\sim_Hv$ implies $w\sim_G v$. Thus, $\CM,v \Vdash \phi$ whenever $w\sim_H v$. Therefore, it follows that $\CM,w \Vdash [H]\phi$.

\item (Equiv): Assume that for an arbitrary $G$, there exists $A\in G$ with $B\subseteq A$. Further, let $w\in W$ be arbitrary such that $\CM,w \Vdash [G \cup \{B\}]\phi$. Therefore, for all $v\in W$, $w\sim_{G \cup \{B\}} v$ implies $\CM,v \Vdash \phi$ by assumption. By $\kfour$ we have that $\sim_G \subseteq \sim_{G\cup \{B\}}$ and thus $w\sim_G v$ implies $w\sim_{G\cup \{B\}}v$. Hence, $\CM,v \Vdash \phi$ whenever $w\sim_G v$ and thus $\CM,w\Vdash [G]\phi$.
\end{enumerate}

Lastly, we show ($G$-Nec). Let $A \in \mathcal{L}$ and assume $\Vdash_\kappa A$. We need to show that $[G]A$ is $\kappa$-valid. Let $\CM = (W,\sim,V)$ be an arbitrary $\kappa$-model. By assumption, $\CM ,w \Vdash A$ for all $w\in W$. Thus, for any $v\in W$ with $w \sim_G v$, it holds that $\CM,v \Vdash A$. By the definition of truth, $\CM, w \Vdash [G]A$, and since $w\in W$ was arbitrary, $\CM \Vdash [G]A$. Moreover, due to $\CM$ being arbitrary, $[G]A$ is $\kappa$-valid.
\end{proof}

In what follows, we set up the machinery to prove completeness of $\CL$ with respect to $\kappa$-models.

\begin{definition}
Let $G$ be an agent pattern and let $\Gamma \subseteq \CL$, we define
\[
\Gamma \setminus [G] := \{\phi \mid [G]\phi\in \Gamma \}.
\]
\end{definition}

\begin{definition}[Consistent set]
A set $\Gamma \subseteq \CL$ is consistent if and only if $\Gamma \not \vdash \bot$. $\Gamma$ is maximal consistent if none of its proper supersets is consistent.
\end{definition}
 
\begin{definition}[Canonical Model]\label{def:cmodel}
The canonical model $\CM^c = (W^c,\sim^c,V^c)$ for $\CL$ is defined as 
\begin{enumerate}
\item $W^c := \{\Gamma \subseteq \CL \mid \Gamma \text{ is a maximal consistent set}\}$ is the set of possible worlds;
\item $\sim^c$ is a function that assigns to each agent pattern $G$ a relation \[
\sim^c_G := \{(\Gamma, \Delta) \in W^c\times W^c \mid \Gamma \setminus [G] \subseteq \Delta\};
\]
\item $V^c : \P \rightarrow \pow(W^c)$ is a function defined by
\[
V^c(p) := \{\Gamma \in W^c \mid p \in \Gamma\}.
\]
\end{enumerate}
\end{definition}

\begin{lemma}\label{lem:alive}
Let $\CM^c = (W^c,\sim^c,V^c)$ be the canonical model, $G$ be an agent pattern, and $\Gamma \in W^c$, then
\[
\Gamma \in \alive(G)_{\CM^c} \quad \text{iff}\quad \lalive(G)\in \Gamma.
\]
\end{lemma}
\begin{proof}
We first show that
\[
\Gamma \in \alive(G)_{\CM^c} \quad \text{implies}\quad \lalive(G)\in \Gamma.
\]
Assume $\Gamma \in \alive(G)_{\CM^c}$. By Definition \ref{def:alive}, $\Gamma \sim^c_G \Gamma$. Towards a contradiction, assume that $\lalive(G) \not\in \Gamma$, i.e.,~$[G]\bot \in \Gamma$ by the maximal consistency of $\Gamma$. Since $\Gamma \sim^c_G \Gamma$,~i.e., $\Gamma \setminus [G]\subseteq \Gamma$, this yields $\bot \in \Gamma$, which contradicts the consistency of $\Gamma$. Thus $\lalive(G) \in \Gamma$.
We now show that
\[
\lalive(G)\in \Gamma \quad \text{implies}\quad \Gamma \in \alive(G)_{\CM^c}.
\]
Assume $\lalive(G)\in \Gamma$. We need to show $\Gamma \sim^c_G\Gamma$. Let $\phi \in \Gamma \setminus [G]$,~i.e.,~$[G]\phi \in \Gamma$. Since 
\[
\lalive(G) \to ([G]\phi \to \phi)
\]
is an axiom of $\synminus$ and by the maximal consistency of $\Gamma$, it follows that$~{\phi \in \Gamma}$, i.e., $\Gamma\sim^c_G\Gamma$. By Definition \ref{def:alive}, we obtain $\Gamma\in \alive(G)_{\CM^c}$.
\end{proof}

\begin{lemma}
$\CM^c$ is a $\kappa$-model.
\end{lemma}
\begin{proof} We show that $\CM^c$ satisfies all properties of $\kappa$-models.
\begin{enumerate}
\item Symmetry: Assume $\Gamma \setminus [G] \subseteq \Delta$. We need to show $\Delta \setminus [G] \subseteq \Gamma$. Let $\phi \in \Delta \setminus [G]$, i.e.,~$[G]\phi \in \Delta$, and assume towards a contradiction that $\phi \not \in \Gamma$. Since $\Gamma$ is a maximal consistent set and (B) is an axiom of $\synminus$, we have $\neg \phi \in \Gamma$ as well as $[G]\neg[G]\neg\phi\in \Gamma$ by the maximal consistency of $\Gamma$. Therefore, $\neg[G]\neg\phi \in \Gamma\setminus [G]$ and thus $\neg[G]\neg\phi \in \Delta$. This is a contradiction because we assumed $[G]\phi \in \Delta$. Hence, we conclude $\phi \in \Gamma$ which shows that $\Delta \setminus [G] \subseteq \Gamma$.

\item Transitivity: Assume $\Gamma \sim^c_G \Delta$ and $ \Delta \sim^c_G \Phi$. We need to show $\Gamma \sim^c_G \Phi$. By assumption, we have $\Gamma \setminus [G] \subseteq \Delta$ and $\Delta \setminus [G] \subseteq \Phi$. Let $\phi \in \Gamma \setminus [G]$, i.e., $[G]\phi \in \Gamma$. Since $\Gamma$ is maximally consistent and because (4) is an axiom of $\synminus$, we have $[G][G]\phi \in \Gamma$, and thus $[G]\phi \in \Delta$. Further, since $\Delta \sim^c_G \Phi$, we have $\phi \in \Phi$ and hence $\Gamma \setminus [G] \subseteq \Phi$.

\item $\kone$: Assume $\Gamma \in \alive(G)$ and $\Gamma \in \alive(H)$. By Lemma \ref{lem:alive} it follows that $\lalive(G) \in \Gamma$ and $\lalive(H) \in \Gamma$. Since $\Gamma$ is a maximal consistent set and (Union) is an axiom of $\synminus$ it follows that $\lalive(G\cup H) \in \Gamma$. By Lemma~\ref{lem:alive} we have $\Gamma \in \alive(G\cup H)$.

\item $\ktwo$: Assume $A,B\in G$ and let $\Gamma \in \alive(G)$. By Lemma \ref{lem:alive} it follows that $\lalive(G) \in \Gamma$. Since $\Gamma$ is a maximal consistent set and (Clo) is an axiom of $\synminus$ it follows that $\lalive(\{ A \cup B\}) \in \Gamma$ by the maximal consistency of $\Gamma$.  By Lemma \ref{lem:alive} we have $\Gamma \in \alive(\{A \cup B\})$.

\item $\kthree$: Assume $(\Gamma,\Delta) \in \sim^c_H$, we need to show that $(\Gamma,\Delta) \in \sim^c_G$. Let $~{\phi \in \Gamma \setminus [G]}$, i.e., $[G]\phi \in \Gamma$. Since $\Gamma$ is maximally consistent and (Mono) is an axiom of $\synminus$ we have that $[H]\phi \in \Gamma$. Since we assumed $\Gamma \setminus [H] \subseteq \Delta$ we have$~{\phi \in \Delta}$.

\item $\kfour$: Assume $(\Gamma,\Delta) \in \sim^c_{G}$, we need to show that $(\Gamma,\Delta) \in \sim^c_{G\cup\{B\}}$. Let $\phi \in \Gamma \setminus [G \cup \{B\}]$, i.e., $[G \cup \{B\}]\phi \in \Gamma$. Since $\Gamma$ is maximally consistent and (Equiv) is an axiom of $\synminus$ we have that $[G]\phi \in \Gamma$. Since we assumed $\Gamma \setminus [G] \subseteq \Delta$ we have $\phi \in \Delta$.

\item $\mathsf{NE}$: Let $\Gamma \in W^c$ be arbitrary. Since $\Gamma$ is maximally consistent and (NE) is an axiom of $\synminus$, there exists an agent pattern $G$ such that $\lalive(G) \in \Gamma$. By Lemma \ref{lem:alive}, $\Gamma \sim^c_G \Gamma$.
\qedhere
\end{enumerate}
\end{proof}

The truth lemma is standard and completeness follows as usual.
\begin{lemma}[Truth lemma]
Let $\CM^c = (W^c,\sim^c,V^c)$ be the canonical model. For each world $\Gamma\in W^c$ and each formula $\phi\in \CL$ we have
\[
\CM^c, \Gamma \Vdash \phi\text{ iff } \phi \in \Gamma.
\]
\end{lemma}

\begin{theorem}[Completeness]\label{thm:kappacomplete}
$\synminus$ is complete with respect to $\kappa$-models.
\end{theorem}

\subsection{$\delta$-models}\label{subsec:dmodel}
As demonstrated in Example \ref{example:no std dk}, the previously introduced $\kappa$-models do not necessarily satisfy \eqref{eq:D}. In this section, we introduce $\delta$-models which satisfy it. Furthermore, we show that $\synminus$ is sound and complete with respect to $\delta$-models.

\begin{definition}[$\delta$-model]\label{def:d model}
A pre-model  $\CM = (W,\sim, V)$ is called a $\delta$-model if and only if 
\begin{enumerate}
\itemktwo $\alive(G)_\CM  \subseteq \alive(\{A\cup B\})_\CM $ for $A,B\in G$;
\itemkthree $\sim_H \subseteq \sim_G$, if $G\subseteq H$;
\itemkfour $\sim_G \subseteq \sim_{G\cup \{B\}}$ if there exists $A\in G$ with $B\subseteq A$;
\itemne for all $w\in W$, there exists an agent pattern $G$ such that $w\sim_G w$;
\itemd $\sim_G = \cap_{B\in G} \sim_{\{B\}}$.
\end{enumerate}
\end{definition}
A formula $\phi$ is $\delta$-valid, denoted by $\Vdash_\delta \phi$, if $\CM \Vdash \phi$ for all $\delta$-models $\CM$. Moreover, observe that $\dsf$ implies $\kone$ and thus, $\delta$-models are $\kappa$-models. This implies that $\synminus$ is sound with respect to $\delta$-models. For a pre-model ${~\CM = (W,\sim,V)}$, we call $F = (W,\sim)$ its frame and sometimes write $\CM = (F,V)$ instead. Furthermore, we say that a frame is a $\kappa$-frame, or a $\delta$-frame, in order to indicate which properties the frame satisfies. In what follows, we will show that $\synminus$ is complete with respect to $\delta$-models.

\begin{definition}[History]\label{def:history}
Let $F  = (W,\sim)$ be a frame. A history $h$ is a non-empty and finite sequence of triples $(w,G,v)$ where 
\begin{enumerate}
\item $w\sim_G v$ and $G$ is maximal under set inclusion. That means, there does not exist an agent pattern $G'$ with $G \subsetneq G'$ and $w \sim_{G'}v$;
\item if $(w',G',v')$ is the successor of $(w,G,v)$, then $v=w'$.
\end{enumerate} 
\end{definition}

We write $\ell(h)$ to denote the last world of a history $h$. That is, if $(w,G,v)$ is the last element of $h$, then $\ell(h) = v$. Furthermore, $h \parallel (\ell(h),G,v)$ denotes the extension of $h$ with $(\ell(h),G,v)$. The set of all histories over a frame $F$ is denoted by $H_F$. 

\begin{lemma}[Downwards closure]\label{lem:closure}
Let $F=(W,\sim)$ be a $\kappa$-frame and consider a history $h\in H_F$ and let $ (w,U,v)$ be an element of $h$. If there exists $A\in U$ with $B\subseteq A$, then $B\in U$.
\end{lemma}
\begin{proof}
Towards a contradiction, suppose that $B \not \in U$. Consider the agent pattern $U' = U\cup \{B\}$. Clearly $U\subsetneq U'$ and $w\sim_{U'}v$ because of $\kfour$. This contradicts $U$ being maximal under set inclusion and thus, $B\in U$.
\end{proof}

\begin{definition}[$\rightarrow_G$]\label{def:arrow_G}
Let $F  = (W,\sim)$ be a frame and consider $h,h'\in H_F$. For an agent pattern $G$, we define $\rightarrow_G \subseteq H_F \times H_F$ as follows
\[ 
h\rightarrow_G h' \quad \text{if{f}} \quad h' = h \parallel (\ell(h),U,\ell(h')) \text{ and } G\subseteq U.
\]
\end{definition}

\begin{lemma}\label{lem:arrow_properties}
Let $F = (W,\sim)$ be a $\kappa$-frame and consider $h,h' \in H_F$ with $~{h\rightarrow_G h'}$ for some agent pattern $G$. The following hold:
\begin{enumerate}
\item if $H \subseteq G$, then $h\rightarrow_H h'$;
\item if there exists $A\in G$ and $B\subseteq A$, then $h\rightarrow_{\{B\}}h'$;
\item if $h\rightarrow_H h'$, then $h\rightarrow_{G\cup H}h'$.
\end{enumerate}
\end{lemma}
\begin{proof}
By assumption, $h' = h\parallel(w,U,v)$ with $G\subseteq U$. For the first case, we have $H\subseteq G \subseteq U$ and thus $h\rightarrow_H h'$ by Definition \ref{def:arrow_G}. For the second case, we have $\{B\}\subseteq U$ by Lemma \ref{lem:closure} and thus $h\rightarrow_{\{B\}}h'$ by Definition \ref{def:arrow_G}. Lastly, assume $h\rightarrow_G h'$ and $h\rightarrow_H h'$, i.e., $G\subseteq U$ and $H\subseteq U$. Hence $G\cup H \subseteq U$ by the properties of set union, and therefore $h\rightarrow_{G\cup H}h'$ by Definition \ref{def:arrow_G}.
\end{proof}

\begin{definition}[$U(F)$]\label{def:u(f)}
Let $F= (W,\sim)$ be a frame. We define the unravelled frame $U(F) = (H_F,\{\approx_G\}_{G\subseteq \pow(\ag)\setminus \{\emptyset\}})$  where
$\approx_G$ is the transitive closure of the symmetric closure of $\rightarrow_G$, i.e., $\approx_G = (\rightarrow_G \cup \rightarrow^{-1}_G)^*$.
\end{definition}

\begin{definition}[$U(\CM)$]\label{def:u(f)-model}
Let $F = (W,\sim)$ be a frame and consider the pre-model $\CM= (F,V)$. We call $U(\CM) = (U(F), L)$ with 
\[
h\in L(p) \quad \text{iff} \quad \ell(h) \in V(p)
\]
the unravelled model of $\CM$.
\end{definition}

\begin{definition}{($R$-path)}
Let $R$ be a relation on a set $X$. A $R$-path from $x_1$ to $x_n$ is a sequence
\[
\tau = (x_1,x_2),(x_2,x_3),\cdots,(x_{n-2},x_{n-1}),(x_{n-1},x_n)
\]
with $(x_i,x_{i+1})\in R$ for $1\leq i\leq n-1$.
\end{definition}

\begin{definition}
Let $R$ be a relation on a set $X$. We define 
\[
R\circ R := \{(x,y)\in X \times X \mid \text{ There exists } z\in X \text{ with } (x,z)\in R \text{ and } (z,y)\in R \}.
\]
We abbreviate $\overbrace{(\cdots (R\circ R) \circ R \cdots )}^{n-1 \text{ times}}\circ R)$ to $R^n$.
\end{definition}

\begin{remark}
Let $R$ be relation on a set $X$, then there is a $R$-path of length $n$ from $a$ to $b$ if and only if $(a,b)\in R^n$.
\end{remark}

\begin{corollary}\label{cor:path}
Let $F$ be a $\kappa$-frame. The following two are equivalent
\begin{enumerate}
\item $(h,h')\in \approx_G$;
\item there exists a $(\rightarrow_G \cup \rightarrow^{-1}_G)$-path $\tau$ from $h$ to $h'$.
\end{enumerate}
\end{corollary}
For brevity we refer to $(\rightarrow_G \cup \rightarrow^{-1}_G)$-paths as $\rightarrow_G$-paths.

\begin{remark}\label{rem:h_reflexive}
If $F  = (W,\sim)$ is a frame, then a $\rightarrow_G$-path from $h$ to $h'$ implies the  existence of a $\sim_G$-path from $\ell(h)$ to $\ell(h')$. It follows by transitivity that $\ell(h) \sim_G \ell(h')$.
\end{remark}

\begin{theorem}
Let $F = (W,\sim)$ be a $\kappa$-frame. Then $U(F)$ is a $\kappa$-frame.
\end{theorem}
\begin{proof}
Observe  that $\approx_G = (\rightarrow_G\cup\rightarrow_G^{-1})^*$ is transitive and symmetric, because the transitive closure of a symmetric relation is transitive and symmetric. Moreover, $\mathsf{NE}$ follows because histories are not empty.
\begin{itemize}
\item $\kone$: Assume $(h,h)\in \approx_G$ and $(h,h)\in \approx_H$. By Remark~\ref{rem:h_reflexive}, $\ell(h) \sim_G \ell(h)$ as well as $\ell(h) \sim_H \ell(h)$. Since $F$ satisfies $\kone$, it holds that $\ell(h) \sim_{G \cup H} \ell(h)$. Hence, $h^* = h\parallel (\ell(h),U,\ell(h))$ with $G\cup H \subseteq U$ is a valid history and $(h,h^*)\in \approx_{G\cup H}$. We obtain $(h,h)\in \approx_{G\cup H}$ by symmetry and transitivity.

\item $\ktwo$: Assume $(h,h)\in \approx_G$ and let $A,B\in G$. By assumption and Remark~\ref{rem:h_reflexive},  $\ell(h)\sim_G \ell(h)$. Because $F$ satisfies $\ktwo$, we have $\ell(h) \sim_{\{A\cup B\}}\ell(h)$. Thus $h^* = h \parallel (\ell(h),U,\ell(h))$ with $\{A\cup B\}\subseteq U$ is a valid history which implies that $(h,h)\in \approx_{\{A\cup B\}}$ by symmetry and transitivity.

\item $\kthree$:  Assume $(h,h')\in \approx_H$ and $G\subseteq H$. By Corollary \ref{cor:path}, there must exist a $\rightarrow_H$-path $\tau$ from $h$ to $h'$. Let $(s,s') \in \tau$ be arbitrary. Since $F$ satisfies $\kthree$, Lemma~\ref{lem:arrow_properties} implies that $s\approx_G s'$. Therefore, $\tau$ is a $\rightarrow_G$-path and $(h,h')\in \approx_G$ by Corollary \ref{cor:path}.

\item $\kfour$:  Assume $(h,h')\in \approx_G$ and that there exists $A\in G$ with $B\subseteq A$. By Corollary \ref{cor:path} there exists a $\rightarrow_G$-path $\tau$ from $h$ to $h'$. Let $(s,s') \in \tau$ be arbitrary. Since $F$ satisfies $\kfour$, Lemma \ref{lem:arrow_properties} implies that $s\approx_{G\cup \{B\}}s'$. Therefore, $\tau$ is a $\rightarrow_{G\cup \{B\}}$-path  and $(h,h')\in \approx_{G\cup \{B\}}$ by Corollary \ref{cor:path}.
\qedhere
\end{itemize}
\end{proof}

We call a $R$-path from $x_1$ to $x_n$ non-redundant if and only if  $x_i \neq x_{i+2}$ for $1 \leq i <n-1$. Further, we often write 
\[
\tau = x_1Rx_2Rx_3 \cdots x_{n-1}Rx_n
\]
instead of $\tau = (x_1,x_2),(x_2,x_3),\cdots,(x_{n-2},x_{n-1}),(x_{n-1},x_n)$.  If $R = S\cup S^{-1}$, we will use $S$ and $S^{-1}$ in a path instead of $R$. We say that a $\rightarrow_G$-path has a \emph{change of direction}, if $(x_i,x_{i+1})\in \rightarrow_G$ (respectively $(x_i,x_{i+1})\in \rightarrow^{-1}_G$) and $(x_{i+1},x_{i+2})\in \rightarrow^{-1}_G$ (respectively $(x_{i+1},x_{i+2})\in \rightarrow_G$). In what follows, we will write $h \leftarrow_G h'$ instead of $h \rightarrow_G^{-1} h'$ for better readability.

\begin{lemma}\label{lem:prefix}
Let $F=(W,\sim)$ be a $\kappa$-frame and consider the unravelled frame $U(F)$. If $(h_1,h_n)\in \approx_G$, then $h_1$ and $h_n$ have a common prefix.
\end{lemma}
\begin{proof}
If one history is a prefix of the other, the claim follows trivially. Hence, we assume that neither of them is a prefix of the other.
Corollary \ref{cor:path} ensures that there exists a $\rightarrow_G$-path $\tau = (h_1,h_2),\cdots,(h_{n-1},h_n)$, which we can assume to be non-redundant. Moreover, due to our first assumption, $\tau$ must have at least one change of direction. In order to show that a common prefix of $h_1$ and $h_n$ exists, we show that $\tau$ has exactly one change of direction, and that change is of the form $h_{i-1} \leftarrow h_i \rightarrow h_{i+1}$ with $i>1$.

\indent Let $h_i$ be the history at which the first change of direction occurs. Notice, that this implies that $i>1$. First, we observe that this change of direction cannot be of the form $h_{i-1} \rightarrow h_i \leftarrow h_{i+1}$, because this would imply $h_{i-1} = h_{i+1}$ and contradict $\tau$ being non-redundant. Therefore, the first change of direction must be of the form $h_{i-1} \leftarrow h_i \rightarrow h_{i+1}$. A subsequent change of direction would be of the form $h_{i+k-1} \rightarrow_G h_{i+k} \leftarrow_G h_{i+k+1}$ for $k\geq 1$ which contradicts that $\tau$ is non-redundant. Therefore, $h_i$ is a common prefix.
\end{proof}

\begin{lemma}\label{lem:idontknow}
Let $F = (W,\sim)$ be a $\kappa$-frame and consider the unravelled frame $U(F)$. For any two agent patterns $G$ and $H$, the following holds
\[
(h,h')\in \approx_G \text{ and } (h,h')\in \approx_H \text{ implies } (h,h')\in \approx_{G\cup H}.
\]
\end{lemma}
\begin{proof}
Assume $h\approx_G h'$ and $h\approx_H h'$. By Lemma \ref{lem:prefix} the following paths exist
\begin{multicols}{2}
\begin{itemize}
\item $h''_G \rightarrow_G \cdots \rightarrow_G h$,
\item $h''_G \rightarrow_G \cdots \rightarrow_G h'$,
\item $h''_H \rightarrow_H \cdots \rightarrow_H h$, and
\item $h''_H \rightarrow_H \cdots \rightarrow_H h'$.
\end{itemize}
\end{multicols}
\noindent
Observe that either $h''_G = h''_H$ or one of the histories is a proper prefix of the other. If $h''_G \neq h''_H$, let $h''$ be the longer history. If they are of the same length, fix either $h'' = h''_G$ or $h'' = h''_H$. We can write $h$ and $h'$ as
\[
h = h'' \parallel (w_1,G_1,w_2) \parallel \cdots \parallel (w_{n-1}, G_n,w_n), \text{ and }
\] 
\[
h' = h'' \parallel (w'_1,G'_1,w'_2) \parallel \cdots \parallel (w'_{m-1} G'_m,w'_m).
\]
By Definition \ref{def:u(f)} we have $G\subseteq G_i,G'_j$ and $H\subseteq G_i,G'_j$ for all $i,j\geq 0$. Thus, by Lemma \ref{lem:arrow_properties} and Corollary \ref{cor:path}, we have $h\approx_{G\cup H}h'$.
\end{proof}

\begin{lemma}\label{lem:std_group_knowledge}
Let $F  = (W,\sim)$ be a $\kappa$-frame. $U(F)$ satisfies $\dsf$.
\end{lemma}
\begin{proof}
$\approx_G \subseteq \bigcap_{B\in G} \approx_{\{B\}}$ follows directly by $\kthree$. For the other direction, let $G = \{A_1,...,A_n\}$ be an agent pattern and consider the sets
\[
B_1 = \{A_1\} \text{ and } B_i = \{A_i\}\cup B_{i-1} \text{ for } 2\leq i\leq n.
\]
If $(u,v)\in \bigcap_{B\in G}\approx_{\{B\}}$, then $(u,v)\in \approx_{A_i}$ for $1\leq i \leq n$. Applying Lemma~\ref{lem:idontknow} inductively yields  $(u,v)\in \approx_{B_i}$,  for $1\leq i\leq n$. Since $B_n = G$, we obtain $(u,v)\in \approx_G$.
\end{proof}

\begin{corollary}
If $F$ is a $\kappa$-frame, then $U(F)$ is a $\delta$-frame.
\end{corollary}

\begin{definition}[Functional bisimulation]
Let $\CM = (W_\CM,\sim^\CM,V_\CM)$, and \linebreak $\CN = (W_\CN,\sim^\CN,V_\CN)$ be two pre-models. $f: W_\CM \rightarrow W_\CN$ is a functional bisimulation if and only if
\begin{enumerate}
\item $\textbf{Atom}$: for all $w\in W_\CM$, $V_\CM(w) = V_\CN(f(w)) $.
\item $\textbf{Forth}$: for any agent pattern $G$, $w\sim^\CM_G v$ implies $f(w)\sim^\CN_G f(v)$.
\item $\textbf{Back}$:  for any agent pattern $G$, $f(w) \sim^\CN_G v'$ implies that there exists $w'$ with $f(w')=v'$ such that $w \sim^\CM_Gw'$.
\end{enumerate}
\end{definition}

\begin{lemma}
If $f: W_\CM \rightarrow W_\CN$ is a functional bisimulation, then for any formula $\varphi$ we have $\CM,w\Vdash \varphi$ if and only if $\CN,f(w) \Vdash \varphi$.
\end{lemma}
\begin{proof}
Let $f$ be a functional bisimulation. We show the claim by induction on the length of $\varphi$.
\begin{enumerate}
\item If $\varphi \in \P$, then the claim follows by $\textbf{Atom}$.
\item $\varphi = \neg \phi$ follows by the induction hypothesis.
\item $\varphi = \phi \land \psi$ follows by the induction hypothesis.
\item We will show equivalently that
\[
\CM,w\not\Vdash [G]\phi \quad \text{iff}\quad \CN,f(w)\not\Vdash [G]\phi.
\]
From left to right, assume $\CM,w\not\Vdash [G]\phi$. Hence, there exists $v\in W_\CM$ with $w\sim^\CM_G v$ and $\CM,v\not\Vdash \phi$. By \textbf{Forth} we have that $f(w)\sim^\CN_G f(v)$ and by the induction hypothesis we have $\CN,f(v)\not\Vdash \phi$. Therefore, it holds that $\CN, f(w) \not \Vdash [G]\phi$. For the other direction, assume $\CN,f(w)\not\Vdash [G]\phi$. Hence, there exists $v'\in W_\CN$ with $f(w) \sim^\CN_G v'$ and $\CN,v'\not\Vdash \phi$. By \textbf{Back} there exists $w'\in W_\CM$ such that $f(w') = v'$ and $w\sim^\CM_G w'$. By the induction hypothesis we obtain $\CM,w'\not\Vdash \phi$ and thus $\CM, w \not \Vdash [G]\phi$.\qedhere
\end{enumerate}
\end{proof}

Let $F = (W,\sim)$ be a $\kappa$-frame and consider a $\kappa$-model $\CM= (W,\sim,V)$ as well as its unravelled model $U(\CM) = (U(F),L)$. It is straightforward to show that the mapping $\mathsf{last}: H_F \rightarrow W$ that maps each $h\in H_F$ to $\ell(h) \in W$ is a functional bisimulation. We only show the claim for $\textbf{Back}$. Assume that $w\in H_F$ and $v'\in W$ such that $\ell(w) \sim_G v'$. It follows that, $w' = w\parallel (\ell(w),U,v')$ with $G \subseteq U$ is a valid history with $w' \approx_G w$ and $\ell(w') = v'$. Finally, the completeness of $\synminus$ with respect $\delta$-models is an immediate consequence.

\begin{theorem}\label{thm:delta_complete}
$\synminus$ is sound and complete with respect to $\delta$-models.
\end{theorem}

\subsection{Completeness of $\syn$}\label{sec:syn_compl}
In this section, we will show that $\syn$ is sound and complete with respect to proper $\delta$-models. 
\begin{definition}
Let $\CM = (W,\sim,V)$ be a pre-model. For $w\in W$, we define
\[
\overline{w} :=  \{B \mid \exists G. B\in G \text{ and } w\in \alive(G)_\CM\}.
\]
\end{definition}

\begin{definition}[$\equiv$]
Let $\CM = (W,\sim,V)$ be a pre-model. We define the relation $\equiv$ on $W \times W$ as
\[
w \equiv v \quad\text{iff}\quad \overline{w}= \overline{v} \text{ and }  w\sim_{\overline{w}} v.
\]
\end{definition}

\begin{definition}[Proper]\label{def:proper}
Let $\CM = (W,\sim,V)$ be a pre-model. We say that $\CM$ is proper if and only if $w \equiv v$ implies $w = v$.
\end{definition}

If $G$ is an agent pattern, we denote the set of maximal elements of $G$ with $\max(G)$. Notice that $\max(\overline{w})$ always contains exactly one element.

\begin{remark}\label{rem:alive_dead}
Let $\CM = (W,\sim,V)$ be a $\kappa$-model. The following can be shown by using the properties of $\kappa$-models and the fact that $\overline{w}$ is maximal under subsets: 
\begin{enumerate}
\item  $\CM,w \Vdash \lalive(\overline{w})$ and $\CM,w \Vdash \dead(\overline{w}^C)$ are always the case;
\item $\CM,w \Vdash \lalive(G) \land \dead(G^C)$ if and only if $\max(G) = \max(\overline{w})$;
\item if $\max(G) = \max(\overline{w})$, then $w\sim_G v$ and $\CM,v \Vdash \dead(G^C)$ imply $\overline{w} = \overline{v}$.
\end{enumerate}
\end{remark}

\begin{definition}\label{def:proper_construction}
Let $\CM = (W,\sim,V)$ be a pre-model. We define the model $~{\CM^\rho = (W^\rho,\sim^\rho, V^\rho)}$ as
\begin{enumerate}
\item $W^\rho = W / \equiv$ is the set of equivalence classes of $\equiv$;
\item $[w]\sim^\rho_G [v]$ if and only if $ w \sim_G v$;
\item for any $p\in \mathsf{Prop}$, $[w]\in V^\rho(p)$ if and only if $w \in V(p)$.
\end{enumerate}
$\CM^\rho$ is well-defined, if, for any two worlds  $w,v\in W$ and $p\in \mathsf{Prop}$,
\[ 
w\equiv v \text{ implies }w\in V(p) \iff v\in V(p).
\]
\end{definition}

\begin{lemma}\label{lem:proper}
Let $\CM = (W,\sim,V)$ be a pre-model such that $\CM^\rho$ is well-defined. We find that:
\begin{enumerate}
\item $\CM^\rho$ is proper;
\item $\CM,w \Vdash \phi \quad \text{if and only if}\quad \CM^\rho, [w] \Vdash \phi$.
\end{enumerate}
\end{lemma}
\begin{proof}
In order to show that $\CM^\rho$ is proper, observe that $[w]\equiv[v]$ implies $w \equiv v$. Indeed, since $[w]\sim^\rho_G [v]$ if and only if $w\sim_G v$, it holds that $\overline{[w]}=\overline{w}$ and $w\sim_{\overline{w}} v$. Hence, $w$ and $v$ belong to the same equivalence class, i.e., $[w] = [v]$.

For the second claim, we show the direction from right to left by induction on the length of $\phi$. The other direction is symmetric. The base case follows because $\CM^\rho$ is well defined. The only case left is $\phi = [G]\psi$. Assume $~{\CM^\rho,[w] \Vdash [G]\psi}$. We need to show $\CM , w \Vdash [G]\psi$. Let $v\in W$ be such that  $w\sim_G v$, i.e., $[w] \sim^\rho_G [v]$. By the definition of truth $\CM^\rho, [v] \Vdash \psi$ and by the induction hypothesis $\CM,v \Vdash \psi$, which concludes the proof.
\end{proof}

Proving soundness and completeness of $\syn$ with respect to proper $\kappa$-models simply requires us to show that proper $\kappa$-models satisfy (P) and that the canonical model is proper.

\begin{lemma}[Soundness]
$\syn$ is sound with respect to proper $\kappa$-models.
\end{lemma}
\begin{proof}
We showed the cases for $\synminus$ in the proof of Theorem~\ref{Thm:sound_kappa}. Hence, we only need to show the case for (P). Let $\CM = (W,\sim,V)$ be an arbitrary proper $\kappa$-model. Assume $\CM,w \Vdash \lalive(G) \land \dead(G^C) \land \phi$. By  Remark~\ref{rem:alive_dead}, we find that $\max(G) = \max(\overline{w})$. This implies that for any $v \in W$ with $w \sim_G v$ such that $\CM, v  \Vdash \dead(G^C)$, we have $\overline{w} = \overline{v}$, i.e., $w\equiv v$ and by the properness of $\CM$, it follows that $w=v$. Therefore $\CM ,v \Vdash \phi$.
\end{proof}

\begin{theorem}
The canonical model $\CM^c$ for $\syn$ is a proper $\kappa$-model.
\end{theorem}
\begin{proof}
We showed the cases for $\synminus$ in the proof of Theorem~\ref{thm:kappacomplete}. Hence, it suffices to show that $\CM^c$ is proper. Let $\Gamma,\Delta \in W^c$ such that $\Gamma \equiv \Delta$, i.e., $\overline{\Gamma} = \overline{\Delta} = G$ and $\Gamma \sim^c_G\Delta$. We now show $\Gamma = \Delta$, i.e., for any $\phi \in \CL$, $\phi \in \Gamma$ if and only if $\phi \in \Delta$. We show the direction from left to right. The other direction is symmetric.  Assume $\phi \in \Gamma$. Since $\Gamma$ is a maximal consistent set, it follows by Remark~\ref{rem:alive_dead} and Lemma~\ref{lem:alive} that $\lalive(G)\land\dead(G^C)\land\phi\in \Gamma$. Furthermore, by (P), we have $[G](\dead(G^C) \rightarrow \phi) \in \Gamma$ and thus $\dead(G^C) \rightarrow \phi\in \Delta$. By assumption, $\dead(G^C)\in \Delta$ and thus $\phi\in \Delta$, because $\Delta$ is a maximal consistent set. 
\end{proof}

\begin{corollary}
$\syn$ is sound and complete with respect to proper $\kappa$-models.
\end{corollary}

Lemma \ref{lem:construction} below shows that 
the construction of Definition~\ref{def:proper_construction} can be applied to the unravelled canonical model. 

\begin{lemma}\label{lem:construction}
Let $\CM^c$ be the canonical model for $\syn$, and let $~{U(\CM^c) = (H,L)}$ be the unravelled canonical model. For any $h,h'\in H$ with $h \equiv h'$ and $p\in\mathsf{Prop}$ it holds that  
\[
h\in L(p) \quad \text{iff}\quad h' \in L(p).
\]
\end{lemma}
\begin{proof}
Let $\CM^c = (W^c,\sim^c,V^c)$ be the canonical model for $\syn$. Furthermore, consider two histories $h,h'$ of $U(\CM^c)$ with $h \equiv h'$, i.e., $ \overline{h} = \overline{h'} = G$ and $h \approx_G h'$. Let $\Gamma = \ell(h)$ and $\Delta = \ell(h')$. By Remark~\ref{rem:h_reflexive}, $\ell(h) \sim_G \ell(h')$, i.e., $\Gamma \sim^c_G \Delta$. We now show the direction from left to right. The other direction is symmetric. Let $p \in \mathsf{Prop}$ with $h \in L(p)$, i.e., $ \Gamma \in V^c(p)$. Since $\Gamma$ is a maximal consistent set it follows by Remark~\ref{rem:alive_dead} and Lemma~\ref{lem:alive} that $
~{\lalive(G)\land \dead(G^c) \land p \in \Gamma}$.
By (P), $[G](\dead(G^C) \rightarrow p) \in \Gamma$. Therefore, $\dead(G^C) \rightarrow p \in \Delta$. Lastly, by assumption and Remark \ref{rem:alive_dead}, we have $\dead(G^C)\in \Delta$ and since $\Delta$ is a maximal consistent set, it follows that $ p\in \Delta$, i.e., $\ell(h') \in V^c(p) $.
\end{proof}

\begin{corollary}
Let $\CM^c$ be the canonical model for $\syn$. It holds that $U(\CM^c)^\rho$ is proper and
\[
U(\CM^c)^\rho , [h] \Vdash \phi \quad\text{iff}\quad U(\CM),h\Vdash \phi.
\]
\end{corollary}

\begin{corollary}\label{thm:soundcompete_properdmodels}
$\syn$ is sound and complete with respect to proper $\delta$-models.
\end{corollary}

\subsection{$\delta$-translations}\label{subsec:translation}
In this section, we show how any proper $\delta$-model can be transformed to an equivalent simplicial model. Completeness of $\syn$ with respect to simplicial models follows immediately.

\begin{definition}[$\delta$-translation]
Let $\CM = (W, \sim, V)$ be a $\delta$-model and consider a simplicial model $\CC = (\BC,L)$. $\CC$ is $\delta$-translation of $\CM$ if and only if there exists a mapping $T: W \rightarrow \BC$ such that for all $w,v\in W$
\begin{enumerate}
\item $G \subseteq (T(w) \cap T(v) )^\circ $ iff $ w\sim_Gv$;
\item $T(w) \in L(p) $ iff $ w \in V(p)$.
\end{enumerate} 
\end{definition}

\begin{algo*}[ht]{}
  \vbox{
  \small
  \begin{numbertabbing}
    xxxx\=xxxx\=xxxx\=xxxx\=xxxx\=xxxx\=MMMMMMMMMMMMMMMMMMM\=\kill
    \textbf{Input} \label{ }\\
    \> \text{A proper $\delta$-frame }$F = (W,\sim)$\label{}\\\\
    \textbf{Initialisation} \label{alg:init:start}\\
    \> $w^*_i := \max \{ A \subseteq \ag \mid w_i \sim_{\{A\}} w_i \}$\label{}\\
    \> $S_i := \{(A,i) \mid A\subseteq w^*_i\}$ for $1\leq i \leq n$\label{alg:init:end}\\\\
    \textbf{Transformation} \label{alg:trans:start} \\ 
    \>$i=1$ \label{ }\\
    \>$j=1$ \label{ }\\
    \> \textbf{while} $w_i$ exists \textbf{do} \label{ }\\
    \>\> $j \gets i+1$ \label{}\\
     \>\> \textbf{while} $w_j$ exists \textbf{do} \label{ }\\
    \>\>\>\textbf{for each} $G\subseteq \pow(\ag) \setminus \{\emptyset\}$ \textbf{with} $w_i \sim_G w_j$\textbf{ do} \label{ }\\
    \>\>\>\>\textbf{for each} $B\in G$ \textbf{ do}\label{}\\
    \>\>\>\>\> $k  \gets \{l \mid w_l \sim_{\{B\}}w_i\}$ \label{alg:smallest}\\
    \>\>\>\>\> $S_j \gets S_j \setminus \{(B,j)\}$ \label{alg:remove}\\
    \>\>\>\>\> $S_j \gets S_j \cup \{(B,k)\}$  \label{alg:add}\\
    \>\>\> $j \gets j+1$ \label{ }\\
    \>\> $i \gets i+1$ \label{ }\\\\
   \textbf{Output} \label{ }\\
  	 \> $\BC = \{S_i \mid w_i \in W\}$ \label{alg:output}
    \end{numbertabbing}
  }
  \caption{$\delta$-translation}\label{alg:T}
  \label{alg:T}
\end{algo*}

Let $\CM = (W,\sim,V)$ be a proper $\delta$-model and $F = (W,\sim)$ be its $\delta$-frame. We assume an arbitrary enumeration of worlds and write $w_i$ for the $i$-th world. In what follows, we denote the simplicial image of $w_i$ under a mapping $T$, i.e.,~$T(w_i)$, with $S_i$. The next lemmas show that there exists a simplicial model $\CC$ that is a $\delta$-translation of $\CM$. Let $w^*_i \subseteq \ag$ be the maximum set of all agents alive in $w_i$.  Construction~\ref{alg:T} on input $F$ first initialises a simplex $S_i = \{(A,i) \mid A\subseteq w^*_i\}$ for each world $w_i$ (lines \ref{alg:init:start} to \ref{alg:init:end}). At this point, no two different simplices $S_i$ and $S_j$ are connected. Throughout the transformation phase (lines \ref{alg:trans:start} to \ref{alg:add}), Construction~\ref{alg:T} glues related simplices together according to the relation $\sim$ of $F$. It does so by iterating through all pairs $(w_i,w_j)$ with $i <  j$, which suffices by symmetry of~$\sim_G$, and checking for all $G$, if $(w_i,w_j) \in  \sim_G$. If so, for each $B\in G$, $(B,j)$ is replaced with $(B,k)\in S_i$, where $k$ is the smallest index such that $w_k \sim_{\{B\}} w_i$. After the replacement, the simplicies $S_i$ and $S_j$ are connected. Example \ref{example:alg:T} shows a simple execution.

\begin{remark}
Construction~\ref{alg:T} could formally be given by a corecursion. However, we think that the presentation as an algorithm makes it easier to understand.
\end{remark}

We claim that the model $\CC = (\BC,L)$, where $\BC$ is the complex returned by Construction~\ref{alg:T} (line~\ref{alg:output}), and $L$ is a valuation such that $S_i \in L(p)$ if and only if $w_i \in V(p)$, is a $\delta$-translation of $\CM$. 

\begin{example}\label{example:alg:T}
Let $\ag = \{a,b\}$ and consider the proper $\delta$-model $\CM =(W,\sim,V)$ depicted in Figure \ref{fig:ex1}. Construction \ref{alg:T} first initialises
\[
S_1 = \begin{Bmatrix}  ab1 \\ a1,b1 \end{Bmatrix} \text{ , }
S_2 = \begin{Bmatrix} ab2 \\ a2,b2\end{Bmatrix} \text{, and } S_3 =\begin{Bmatrix}
ab3 \\ a3,b3
\end{Bmatrix}.
\]
During the transformation phase, when $i=1$ and $j=2$, Construction \ref{alg:T}  replaces $(a,2)$ and $(b,2)$ with $(a,1)$ and $(b,1)$ due to $w_1 \sim_{\{a\}} w_2$ and $w_1 \sim_{\{b\}} w_2$. Moreover, it replaces $(b,3) \in S_3$ with $(b,1)$. Observe that if Line \ref{alg:smallest} was missing, i.e., $k=i$, then, for $i=2$ and $j=3$, the Construction would add $(b,2)$ to $S_3$ as well, which would make it an ill-formed simplex. The resulting complex is
\[
\BC = \left\{ \begin{Bmatrix} ab1 \\ a1,b1 \end{Bmatrix}, \begin{Bmatrix} ab2 \\ a1,b1 \end{Bmatrix}, \begin{Bmatrix} ab3 \\ a3,b1\end{Bmatrix} \right\}.
\]
\end{example}

\begin{figure}[ht]
\begin{center}
\begin{tikzpicture}[modal]
  \node[world] (w) {$w_1$};
  \node[world] (v) [right=of w] {$w_2$};
  \node[world] (u) [right=of v] {$w_3$};
  \path[-] (w) edge[reflexive] node[above] {$\begin{matrix} ab, a,b \end{matrix}$} (w);
  \path[-] (v) edge[reflexive] node[above] {$\begin{matrix}
  ab,a,b
  \end{matrix}$} (v);
  \path[-] (u) edge[reflexive] node[above] {$\begin{matrix}
  ab,a,b
  \end{matrix}$} (u);
  \path[-] (w) edge node[above] {$\begin{matrix}
  a\\b
  \end{matrix}$} (v);
  \path[-] (v) edge node[above] {$b$} (u);
  \path[-] (w) edge[bend right] node[below] {$b$} (u);
\end{tikzpicture}
\end{center}
\caption{The $\delta$-model for Example \ref{example:alg:T}. Only arrows for maximal agent patterns are shown.}\label{fig:ex1}
\end{figure}
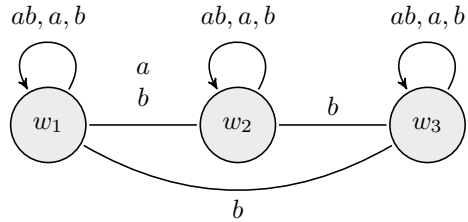

The next Lemma ensures that we can safely assume the existence of an unique and non-empty $S_i$ for each $w_i$ throughout our proofs.

\begin{lemma}[Uniqueness]
Let $F = (W,\sim)$ be a proper $\delta$-frame. After the initialisation of Construction~\ref{alg:T} on input $F$ and after each execution of Line~\ref{alg:add}, the following holds for all $w_i,w_j \in W$:
\begin{enumerate}
\item $S_i \neq \emptyset$;
\item $S_i = S_j$ if and only if $i = j$.
\end{enumerate}
\end{lemma}
\begin{proof}
$S_i \neq \emptyset$ follows immediately because $\delta$-frames satisfy $\mathsf{NE}$ and thus, each $S_i$ is initialised to some non-empty set. Furthermore, since each element that is removed gets replaced, we conclude that $S_i \neq \emptyset$ for all $w_i \in W$ after Line \ref{alg:add}. For the second claim, the direction from right to left follows immediately. For the other direction, observe that if $S_i = S_j$, then $\overline{w_i}= \overline{w_j}$ and $w_i \sim_{\overline{w_i}}w_j$, i.e., $~{w_i \equiv w_j}$. By the properness of $F$, it follows that $w_i = w_j$, i.e., $i = j$.
\end{proof}

Since the smaller index gets precedence, some elements of a simplex may never be replaced. For example, the simplex $S_1$ remains unchanged throughout Construction \ref{alg:T}. Lemma \ref{lem:init} shows that those simplices are correct.

\begin{lemma}\label{lem:init}
Let $F = (W,\sim)$ be a proper $\delta$-frame and consider Construction \ref{alg:T} on input $F$. After the initialisation (lines \ref{alg:init:start} to \ref{alg:init:end}), for all $w_j \in W$, it holds that $(A,j) \in S_j $ iff  $ w_j\sim_{\{A\}}w_j$.
\end{lemma}
\begin{proof}
From left to right, we have $(A,j) \in S_j$ if and only if $A\subseteq w^*_j$. Since $w_j \sim_{\{w^*_j\}}w_j$ by definition, we obtain $w_j\sim_{\{A\}}w_j$ by Lemma \ref{lem:k4}. Regarding the other direction, assume $w_j\sim_{\{A\}}w_j$. By the definition of $w^*_j$ it follows that $A\subseteq w^*_j$ which implies that $(A,j)\in S_j$ after the initialisation phase. 
\end{proof}

\begin{lemma}\label{lem:swap}
Let $F = (W,\sim)$ be a proper $\delta$-frame and  consider $i,j,k \in \mathbb{N}$ such that $i\neq j, i \neq k$, and $ j \neq k $. There does not exist $B \subseteq \ag$ such that, while running Construction \ref{alg:T} on input $F$, $(B,i)$ is replaced with $(B,j)$ and $(B,j)$ is replaced with the pair $(B,k)$.
\end{lemma}
\begin{proof}
Towards a contradiction, assume that there exists such a $B \subseteq \ag$, i.e., for some $S_i$, $(B,i)$ is replaced with $(B,j)$ and $(B,j)$ is replaced with $(B,k)$. By assumption, $i<j<k$ because Construction \ref{alg:T} replaces $(B,l)$ with $(B,m)$ only if $l > m$. Since lines \ref{alg:remove} and \ref{alg:add} are executed at least twice, there exist $G,G' \subseteq \pow(\ag) \setminus \{\emptyset\}$ such that $B\in G$ and $B\in G'$ with $w_i \sim_G w_j$ and $w_j \sim_{G'}w_k$.
Since $\{B\} \subseteq G$ and $\{B\} \subseteq G'$, we have by $\kthree$ that $w_i \sim_{\{B\} } w_j$ and $w_j \sim_{\{B\} }w_k$. Further, $w_i \sim_{\{B\}}w_k$ follows by transitivity of $\sim_{\{B\}}$. But this means, that $i=k$ because $(B,k)$ was replaced by $(B,i)$ prior which contradicts that $i<k$.
\end{proof}

\begin{lemma}\label{lem:t1}
Let $F = (W,\sim)$ be a proper $\delta$-frame. The output $\BC$ of Construction~\ref{alg:T} on input $F$ satisfies the following two properties 
\begin{enumerate}[$\mathsf{T}$1:]
\item Let $S \in \BC$. If $(A,j)\in S$ and $(A,k)\in S$, then $j=k$.
\item Let $w_i,w_j\in W$ and $G\subseteq \pow(\ag) \setminus \{\emptyset\}$, then
\begin{equation*}
\forall B\in G.\exists k\in\mathbb{N}. (B,k) \in S_i \land (B,k)\in S_j \quad\text{iff}\quad w_i\sim_Gw_j.
\end{equation*}
\end{enumerate}
\end{lemma}

\begin{proof} 
$\ti$ follows by construction. Regarding $\tii$, we start by showing the direction from left to right. Assume
\[
\forall B\in G.\exists k\in\mathbb{N}. (B,k) \in S_i  \land (B,k)\in S_j.
\]
Without loss of generality, we fix $i\leq j$. We have that $ (B,k)\in S_i\cap S_j$ only if there exists $G'$ with $w_i \sim_{G'}w_j$ and $B\in G'$. By $\kthree$ we get $w_i \sim_{\{B\}}w_j$. Since $B$ is arbitrary, $\mathsf{D}$ implies that $w_i\sim_G w_j$. 

For the other direction, let $w_i \sim_G w_j$ and let $B\in G$. By construction, $(B,j)\in S_j$ is replaced by $(B,k)$ with $k\leq i$. Let $k$ be the smallest index such that $w_k\sim_G w_i$, i.e.,~if $k<i$, then $(B,i)\in S_i$ was replaced with $(B,k)$ before. Lemma \ref{lem:init} ensures that $(B,k)\in S_k$. Hence, if $k<i$, $(B,j) \in S_j$ and $(B,i)\in S_i$ are both replaced by $(B,k)$ in the same iteration. If $k = i$, only $(B,j)$ is replaced by $(B,i)$, and $(B,i)\in S_i$ at the end of the Construction due to the symmetry of $\sim_G$ and $k$ being minimal. By Lemma \ref{lem:swap} there cannot be any further replacements and thus $(B,k)\in S_i$ and $(B,k)\in S_j$ at the end of Construction \ref{alg:T}.\qedhere
\end{proof}

A consequence of $\tii$ is that

\begin{equation}\label{eq:equiv_indist}
S_i\sim_G S_j\text{ if and only if } w_i\sim_Gw_j.
\end{equation}
\begin{lemma} \label{lem:t2}
Let $F = (W,\sim)$ be a proper $\delta$-frame. Let $\BC$ be the complex returned by Construction \ref{alg:T} on input $F$. For each $w_i\in W$,  $S_i \in \BC$ is a  well formed simplex.
\end{lemma}
\begin{proof}
We show that $S_i$ satisfies $\si,\sii$, and $\siii$.

\begin{enumerate}
\item $S_i$ satisfies $\si$. Let $(A,j) \in S_i$ and $(B,k) \in S_i$ be two distinct maximal elements. First, we note that $A \neq B$ due to $\ti$. By $\tii$, $w_i \in \alive(\{A\})$ and $w_i \in \alive(\{B\})$. By $\kthree$, $w_i \in \alive(\{A\}\cup \{B\})$. Further, by $\ktwo$, $w_i \in \alive(\{A\cup B\})$. Hence, $(A\cup B,l)\in S_i$ for some $l$. But since $A \subsetneq A \cup B$, this contradicts that $(A,i)$ is the maximal element of $S_i$.

\item $S_i$ satisfies $\sii$. Let $(B,j)\in S_i$ and $C\subseteq B$, we need to show that there exists a unique $k$ such that $(C,k) \in S_i$. By Lemma \ref{lem:k4} we have $\sim_{\{B\}} \subseteq \sim_{\{C\}}$. Since we assume that $(B,j)\in S_i$, we get $w_i \sim_{\{B\}}w_i$ by $\tii$. Due to $\sim_{\{B\}} \subseteq \sim_{\{C\}}$,  we have that $w_i \sim_{\{C\}} w_i$. Thus, by $\tii$, there exists $k\in \mathbb{N}$ such that $(C,k) \in S_i$. Condition $\ti$ ensures that $k$ is unique.

\item $S_i$ satisfies $\siii$. Let $\max(S_i) = (B,k)$ and suppose that $(A,j) \in S_i$ for some $A\nsubseteq B$. We have $w_i \sim_{\{A\}} w_i$ and $w_i \sim_{\{B\}} w_i$ by $\tii$. By $\kthree$, it follows that $w_i \in \alive(\{A\}\cup \{B\})$ and by $\ktwo$ we have that $w_i \in \alive(\{A\cup B\})$. Thus $(A\cup B,l) \in S_i$ for some $l$ which contradicts the maximality of $(B,k)$. \qedhere
\end{enumerate}
\end{proof}

\begin{lemma}\label{lem:complex}
Let $F = (W,\sim)$ be a proper $\delta$-frame and consider Construction \ref{alg:T} on input~$F$. The output set $\BC = \{S_i\mid w_i\in W\}$ is a complex.
\end{lemma}
\begin{proof}
In order to show that $\BC$ is a complex, we need to show that it satisfies Condition $\mathsf{C}$. Consider the simplices $S_m, S_n \in \BC$ and assume that there exists $A\subseteq \ag$ and $i$ with $(A,i)\in S_m$ and $(A,i)\in S_n$. We need to show that for all $B\subseteq A$ and all $j\in \BN$ it holds that 
\[
(B,j)\in S_m \quad \text{if{f}} \quad (B,j)\in S_n.
\]
Since $S_m$ and $S_n$ are arbitrary, it is enough to show only one direction.
Assume that $(B,j)\in S_m$. By $\tii$, we have $w_m \sim_{\{A\}} w_n$ and by Lemma \ref{lem:k4}, it holds that $w_m \sim_{\{B\}} w_n$. Since $(B,j)\in S_m$, $\tii$ implies $w_m \sim_{\{B\}} w_j$. Further, by symmetry and transitivity we obtain $w_j \sim_{\{B\}} w_n$. Hence, it holds that $j\leq n$. Therefore, $(B,j)$ replaced $(B,n)$ in $S_n$. By Lemma \ref{lem:swap}, no more replacements of that pair can happen during the execution of Construction \ref{alg:T} and we conclude $(B,j)\in S_n$. 
\end{proof}

\begin{theorem}\label{lem:existence}
For each proper $\delta$-model, there exists a $\delta$-translation.
\end{theorem}
\begin{proof}
Let $\CM = (W,\sim,V)$ be a proper $\delta$-model. Further, consider $\CC = (\BC,L)$, where $\BC$ is the output of Construction \ref{alg:T} on input $(W,\sim)$ and  $L: \BC \rightarrow \P$ is a valuation such that
\begin{equation}\label{eq:equiv_label}
S_i \in L(p) \quad\text{iff}\quad w_i \in V(p).
\end{equation}
By Lemma \ref{lem:complex}, $\CC$ is a simplicial model.
It follows from \eqref{eq:equiv_indist} and \eqref{eq:equiv_label} that $\CC$ is a $\delta$-translation of $\CM$.
\end{proof}

\begin{theorem}\label{thm:translation}
Let $\CM = (W,\sim,V)$ be a proper $\delta$-model and let $\CC = (\BC,L)$ be a $\delta$-translation of $\CM$. It holds that
\[
\CM,w_i  \Vdash \phi \text{ if and only if } \CC,S_i \Vdash_\sigma \phi.
\]
\end{theorem}

\begin{proof}
By induction on the length of formulas.
\begin{enumerate}
\item Let $\phi \equiv \lalive(G)$ for some $G\subseteq \pow(\ag) \setminus \{\emptyset \}$. We have $\CM,w_i \Vdash \lalive(G)$ iff $w_i\sim_Gw_i$ iff $S_i\sim_G S_i$ iff $\CC,S_i \Vdash_\sigma \lalive(G)$.
\item Let $\phi \in \P$. We have $\CM,w_i \Vdash \phi$ iff $w_i\in V(p)$ iff $S_i\in L(p)$ iff $\CC,S_i \Vdash_\sigma \phi$ (by the definition of $L$).
\item Let $\phi \equiv \neg \psi$. Follows by the induction hypothesis.
\item Let $\phi \equiv \psi \land \varphi$. Follows by the induction hypothesis.
\item Let $\phi \equiv [G]\psi$. We equivalently show 
\[
\CM,w_i \not \Vdash [G]\psi \text{ iff }\CC,S_i \not \Vdash_\sigma [G]\psi.
\]
It holds that $\CM,w_i \not \Vdash [G]\psi$ iff there exists $w_j\in W$ with  $w_i\sim_Gw_j$ and $\CM,w_j \not \Vdash \psi$ iff  $G\subseteq (S_i\cap S_j)^\circ$ \eqref{eq:equiv_indist} and $\CC,S_j \not \Vdash_\sigma \psi$ (by hypothesis) iff $\CC,S_i \not \Vdash_\sigma [G]\psi$ by definition.\qedhere
\end{enumerate}
\end{proof}

Hence, if $\not \vdash \varphi$, then there exits a proper $\delta$-model $\CM$ (Corollary \ref{thm:soundcompete_properdmodels}) such that $\CM \not \Vdash \varphi$. By Theorem \ref{thm:translation}, there must exist a simplicial model $\CC$ such that $\CC \not \Vdash_\sigma \varphi$. Thus, $\syn$ is sound and complete with respect to simplicial models.

\begin{corollary}
$\Vdash_\sigma\varphi$ if and only if $\vdash \varphi$.
\end{corollary}

\section{Conclusion}\label{sec:conclusion}
This work presented the normal modal logic $\syn$, a logic in which we can reason about what a group of agents knows beyond their pooled individual knowledge. The new  modality $[G]$ allows us to express the knowledge of a group while also accounting for relations between its members. Hence, the knowledge of two seemingly equal groups may differ due to different relationships among their agents. Moreover, we provided simplicial semantics for $\syn$ based on semi-simplicial sets and demonstrated its use by analysing consensus objects and the dining cryptographers problem. To round off our work on synergistic knowledge, we showed that $\syn$ is sound and complete with respect to our simplicial models. 

Semi-simplicial sets were initially studied in order to represent improper Kripke models and to explore new types of group knowledge (cf. van Ditmarsch et al.~\cite{10.1007/978-3-030-75267-5_1}). While synergistic knowledge definitely is a new form of distributed knowledge, we still must  require that $\syn$ contains the axiom (P) because for any simplex $S$, $S^\circ$ is a valid agent pattern. Thus, $\synminus$ can describe improper Kripke models. For example, the $\delta$-model shown in Figure~\ref{fig:improper} is a valid model for $\synminus$, but not for $\syn$. Applying Construction~\ref{alg:T} would result in a complex with only one world, which means that this complex cannot be a $\delta$-translation of the model shown in Figure~\ref{fig:improper}.
However, if we restrict agent patterns to sets, we might be able to represent improper Kripke models. For example, when treating agent patterns as sets, the complex depicted in Figure~\ref{fig:intro_queue} corresponds to an improper Kripke model.

\begin{figure}[ht]
\begin{center}
\begin{tikzpicture}[modal]
  \node[world] (w) {$w_1$};
  \node[world] (v) [right=of w] {$w_2$};
  \path[-] (w) edge[reflexive] node[above] {$\begin{matrix} ab, a,b \end{matrix}$} (w);
  \path[-] (v) edge[reflexive] node[above] {$\begin{matrix} ab, a,b \end{matrix}$} (v);
  \path[-] (w) edge node[above] {$\begin{matrix}
  ab,a,b
  \end{matrix}$} (v);
\end{tikzpicture}
\end{center}
\caption{An improper $\delta$-model. Only the maximal agent patterns are shown.}\label{fig:improper}
\end{figure}
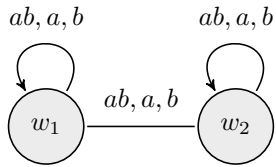

\section{Future Work}\label{sec:future work}
In this work, we presented a candidate for an indistinguishability relation. However, there may be other reasonable options. For example, another possible notion of indistinguishability is one that takes the connectivity of an agent pattern into account. It is an intuitive option for settings where adding an unrelated agent to a pattern must not result in new knowledge. For example, adding an agent without any communication links to a network does not strengthen the network's reasoning power. In other words, an agent pattern is only as strong as its weakest connected component. Unfortunately, the concept of indistinguishability so far cannot capture this. To see why this is the case, consider the model depicted in Figure~\ref{intro:fig1}, where the pattern $G = \{a,b\}$ cannot distinguish $X$ and $Y$. But if we add the completely unrelated agent $c$ to $G$, then, the updated pattern can suddenly distinguish $X$ and $Y$. Hence, we would like a definition in which $\{a,b,c\}$ cannot distinguish $X$ from~$Y$.
Let $G$ be the pattern $\{a, bc\}$. We could interpret this as \emph{$b$ and $c$ communicate with each other} but \emph{there is no communication between $a$ and $b$ or $c$}.
A formula $[G]\phi$ will thus be interpreted as \emph{$a$ knows $\phi$ \textbf{and} the group $b,c$ has distributed knowledge of $\phi$}. Hence, $\{a,b,c\}$ cannot tell the difference between $X$ and $Y$ because $a$ and $b$ cannot do so individually. Moreover, $c$ does not have the means to tell $a$ or $b$ why it can distinguish the two worlds. A further investigation of this kind of indistinguishability seems to be promising for network topologies. We think that it is certainly fruitful to have communication pattern logic (cf.~Casta{\~{n}}eda et al.~\cite{Castaneda2023-CASCPL}) in mind when studying this notion of indistinguishability. A first analysis of this topic can be found in the preliminary version of this work~\cite{DBLP:conf/sss/CachinLS23}.

\begin{figure}[ht]
\begin{center}
\begin{tikzpicture}[scale=0.6]
\node[circle, draw] at (0,0) (w1) [label = below:$ $]{$a$};
\node[circle, draw] at (0,4) (b1)[label = above:$ $]{$b$};
\node[circle, draw] at (-3,2) (g1) [label = left:$ $]{$c$};
\node[circle, draw] at (3,2) (g2)[label = right:$ $]{$c'$};
\node[circle, fill = white] at (-1.5,2) (l1) {$X$};
\node[circle, fill = white] at (1.5,2) (l12) {$Y$};
\draw[-] (w1) to[out=110, in=-110]  (b1);
\draw[-
] (w1) to[out=65, in=-65] (b1);
\path[-] (g1) edge[]  (b1);
\draw[-] (g1) to (w1);
\draw[-] (g2) to (b1);
\draw[-] (g2) to (w1););
\end{tikzpicture}
\caption{A model in which two agents $a$ and $b$ can together distinguish between the worlds $X$ and~$Y$. However, they cannot do so individually.}\label{intro:fig1}
\end{center}
\end{figure}
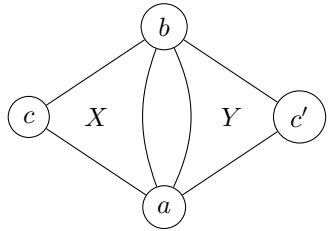

Additionally, we would like to explore the notion of synergy in distributed computing. Especially for tasks which require agents to commit to certain roles. For example, a smart contract that assigns buyers to sellers can be seen as a synergistic primitive for two parties. Indeed, without a buyer, the seller cannot sell and without the seller, the buyer cannot buy. Hence, this task needs two roles: a seller and a buyer. The model depicted in Figure \ref{example:buyerseller} on the left shows how simplicial models can capture this. The buyer $B$ offers an amount $o$ and the seller $S$ sells for a price $p$. Individually, they do not know if a trade can happen. That is, they do not know whether $p\leq o$ (i.e., deal) or $p>o$ (i.e., no deal). Only after querying the previously mentioned smart contract, they know if an exchange can take place. Lastly, it is also reasonable to analyse this from a dynamic standpoint in which querying the smart contract eliminates multiple edges. Figure~\ref{example:buyerseller} shows an update for the model on the left, in  the case that $p\leq o$. It would be interesting to formulate updates that eventually transform models with parallel faces to the usual simplicial complexes. However, it is currently not clear how such update rules could transfer to distributed computing.

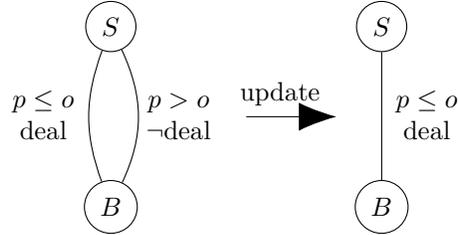
\begin{figure}[ht]\label{example:buyerseller}
\begin{center}
\begin{tikzpicture}[scale=0.6]
\node[circle, draw] at (0,0) (w1) [label = below:$ $]{$B$};
\node[circle, draw]at (0,4) (b1)[label = above:$ $]{$S$};
\node[circle, draw] at (6,0) (w2) [label = below:$ $]{$B$};
\node[circle, draw]at (6,4) (b2)[label = above:$ $]{$S$};
\node[circle, fill = white] at (3.75,2.5) (l1)[label = above:$ $]{update};
\node[circle, fill = white] at (-1.5,2) (l1) {$\begin{matrix}
p \leq o\\
\text{deal}
\end{matrix} $};
\node[circle, fill = white] at (1.5,2) (l12) {$\begin{matrix}
p > o\\
\neg \text{deal}
\end{matrix}$};
\node[circle, fill = white] at (7,2) (l1) {$\begin{matrix}
p \leq o\\
\text{deal}
\end{matrix} $};
\draw[-] (w1) to[out=110, in=-110]  (b1);
\draw[-] (w1) to[out=65, in=-65] (b1);
\draw[-] (w2) to (b2);
 \draw [-{Latex[length=5mm]}] (3,2) -- (5,2);
\end{tikzpicture}
\caption{Without accessing a buyer-seller smart contract, the buyer $B$ and the seller $S$ do not know whether a trade can happen. We can represent the model after accessing the smart contract by simply removing the parallel edge, i.e., global state, that is not the case.}\label{fig:intro}
\end{center}
\end{figure}

\bibliography{references}
\bibliographystyle{elsarticle-harv}
\end{document}